\documentclass[letterpaper, 11pt, singlespacing]{article}       

\advance \topmargin by -\headheight  

\usepackage{amssymb}
\usepackage{xspace}
\usepackage{graphicx}
\usepackage{subfigure}
\usepackage{xspace}
\usepackage{amsmath}
\usepackage{algorithm}
\usepackage{algorithmic}
\usepackage{multirow}

\newcommand{\paren}[1]          {\left( #1 \right)}

\newcommand{\card}[1]           {\left| #1\right|}
\newcommand{\floor}[1]          {\left\lfloor #1 \right\rfloor}
\newcommand{\dt}                {\Delta t}
\newcommand{\ra}[1]             {\overrightarrow{#1}}
\newcommand{\la}[1]             {\overleftarrow{#1}}

\newcommand{\OF}[1]{\ensuremath{\operatorname{\mathsf{#1}}}}

\newcommand{\A}{\OF{A}\xspace}
\newcommand{\B}{\OF{B}\xspace}
\newcommand{\OPT}{\OF{OPT}\xspace}

\newcommand{\NEQUI}{\OF{N-EQUI}\xspace}
\newcommand{\UCEQ}{\OF{U-CEQ}\xspace}
\newcommand{\NE}{\OF{NE}\xspace}
\newcommand{\UC}{\OF{UC}\xspace}
\newcommand{\PF}{\OF{WCEP}\xspace}

\newcommand{\MULTI}{\OF{MultiLaps}\xspace}

\newcommand{\SJF}{\OF{SJF}\xspace}

\newcommand{\algref}[1]         {Algorithm~\ref{alg:#1}}
\newcommand{\secref}[1]         {Section~\ref{sec:#1}}

\newcommand{\tabref}[1]         {Table~\ref{tab:#1}}

\newcommand{\thmref}[1]         {Theorem~\ref{thm:#1}}
\newcommand{\sceref}[1]         {Scenario~\ref{sce:#1}}
\newcommand{\thmreftwo}[2]         {Theorems~\ref{thm:#1} and~\ref{thm:#2}}

\newcommand{\lemref}[1]         {Lemma~\ref{lem:#1}}

\newcommand{\lemreftwo}[2]         {Lemmas~\ref{lem:#1} and~\ref{lem:#2}}
\newcommand{\lemrefthree}[3]         {Lemmas~\ref{lem:#1},~\ref{lem:#2} and~\ref{lem:#3}}
\renewcommand{\eqref}[1]          {Equation~(\ref{eq:#1})}
\newcommand{\eqreff}[1]          {Equations~(\ref{eq:#1})}
\newcommand{\eqreftwo}[2]          {Equations~(\ref{eq:#1}) and~(\ref{eq:#2})}
\newcommand{\ineqref}[1]        {Inequality~(\ref{ineq:#1})}
\newcommand{\ineqreff}[1]        {Inequalities~(\ref{ineq:#1})}

\newcommand{\calset}[1]{\ensuremath{\mathop{\mathcal{#1}}\nolimits}}
\newcommand{\jobset}{\calset{J}}
\newtheorem{theorem}            {Theorem}
\newtheorem{lemma}     {Lemma}
\newtheorem{scenario}     {Scenario}

\newenvironment{proof}      {\noindent{\em Proof.}\hspace{1em}}{\qed}

\def\squarebox#1{\hbox to #1{\hfill\vbox to #1{\vfill}}}
\newcommand{\qedbox}            {\vbox{\hrule\hbox{\vrule\squarebox{.667em}\vrule}\hrule}}
\newcommand{\qed}               {\nopagebreak\mbox{}\hfill\qedbox\smallskip}

\begin{document}

\title{Energy-Efficient Multiprocessor Scheduling for Flow Time and Makespan}

\author{Hongyang Sun\thanks{School of Computer Engineering, Nanyang Technological University, Singapore. {\tt sunh0007@ntu.edu.sg}} \and Yuxiong He\thanks{Microsoft Research, Redmond, WA, USA. {\tt yuxhe@microsoft.com}} \and Wen-Jing Hsu\thanks{School of Computer Engineering, Nanyang Technological University, Singapore. {\tt hsu@ntu.edu.sg}} \and Rui Fan\thanks{School of Computer Engineering, Nanyang Technological University, Singapore. {\tt fanrui@ntu.edu.sg}}}

\date{}

\maketitle

\begin{abstract}
We consider energy-efficient scheduling on multiprocessors, where the speed of each processor can
be individually scaled, and a processor consumes power $s^{\alpha}$ when running at speed $s$, for
$\alpha>1$. A scheduling algorithm needs to decide at any time both processor allocations and
processor speeds for a set of parallel jobs with time-varying parallelism. The objective is to
minimize the sum of the total energy consumption and certain performance metric, which in this
paper includes total flow time and makespan. For both objectives, we present instantaneous
parallelism-clairvoyant (IP-clairvoyant) algorithms that are aware of the instantaneous parallelism
of the jobs at any time but not their future characteristics, such as remaining parallelism and
work. For total flow time plus energy, we present an $O(1)$-competitive algorithm, which
significantly improves upon the best known non-clairvoyant algorithm and is the first constant
competitive result on multiprocessor speed scaling for parallel jobs. In the case of makespan plus
energy, which is considered for the first time in the literature, we present an
$O(\ln^{1-1/\alpha}P)$-competitive algorithm, where $P$ is the total number of processors. We show
that this algorithm is asymptotically optimal by providing a matching lower bound. In addition, we
also study non-clairvoyant scheduling for total flow time plus energy, and present an algorithm
that achieves $O(\ln P)$-competitive for jobs with arbitrary release time and
$O(\ln^{1/\alpha}P)$-competitive for jobs with identical release time.  Finally, we prove an
$\Omega(\ln^{1/\alpha}P)$ lower bound on the competitive ratio of any non-clairvoyant algorithm,
matching the upper bound of our algorithm for jobs with identical release time.
\\

\noindent \textbf{Keywords:} Multiprocessors, Online Scheduling, Dynamic speed scaling,
Energy-performance tradeoff, Competitive Analysis, Total flow time, Makespan
\end{abstract}

\section{Introduction}

Energy has been widely recognized as a key consideration in the design of mobile and
high-performance computing systems.  One popular approach to controlling the energy consumption is
by dynamically varying the speeds of the processors, a technique generally known as \emph{dynamic
speed scaling} \cite{BrooksBoSc00,GrunwaldMoLe01,WeiserWeDe94}. Major chip manufacturers, such as
Intel, AMD and IBM, have produced chips that enable the operating systems to perform dynamic power
management using this technology. It has been observed that, for most CMOS-based processors, the
dynamic power consumption satisfies the \emph{cube-root} rule; that is, the power consumption of a
processor is proportional to $s^3$ when it runs at speed $s$ \cite{BrooksBoSc00,Mudge01}. Since the
seminal paper by Yao, Demers and Shenker \cite{YaoDeSh95}, who initiated the theoretical
investigation of energy-efficient scheduling, most algorithmic researchers have assumed a more
general power function of $s^\alpha$, where $\alpha > 1$ is called the \emph{power parameter}. As
this power function is strictly convex, using dynamic speed scaling can result in a non-linear
tradeoff between the energy consumption and the performance, and this has led to many interesting
new research problems. One challenging problem concerns how to balance the conflicting objectives
of low energy and high performance. The problem has attracted much attention among the algorithmic
community and has become an active research topic in recent years. (See \cite{Albers09,IraniPr05}
for two surveys of the field.)

In this paper, we study the challenging problem of scheduling parallel jobs on multiprocessors for
the energy-performance tradeoff. We focus on systems with per-processor speed scaling capability;
that is, the speed of each processor can be individually scaled
\cite{HerbertMa07,ZhangShDw10,ZhaoJa11}. This kind of architecture has been made possible by the
recent advancements in chip design technology, such as the on-chip switching regulators
\cite{KimGuWe08,KimBrWe11}. Under this setting, a scheduling algorithm needs to have both a
\emph{processor allocation policy}, which determines the number of processors allocated to each
job, and a \emph{speed scaling policy}, which determines the speed of each allocated processor.
Moreover, we assume that the parallel jobs can have time-varying parallelism in different phases of
their executions \cite{Edmonds00,ChanEdPr11,SunCaHs09b}. This poses an additional challenge
compared to scheduling sequential jobs. In particular, it requires a scheduling algorithm to have
dynamic policies in order to respond to the jobs' different resource requirements over time. If not
designed properly, however, the algorithm could waste a large amount of energy or cause severe
execution delays and hence performance degradations.

Our objective is to minimize a linear combination of energy consumption and certain performance
metric, which in this paper includes total flow time and makespan. The \emph{flow time} of a job is
the duration between its release time and completion, and the \emph{total flow time} is the sum of
the flow time of all the jobs in the system. The \emph{makespan} is the largest completion time of
the jobs. Both total flow time and makespan are widely used performance metrics: The former
measures the average response time of all users in the system, and the latter is closely related to
the throughput of the system. Although energy and flow time (or makespan) have different units,
optimizing a linear combination of the two has a natural interpretation if we consider a user who
is willing to spend one unit of energy in order to reduce $\rho$ units of total flow time (or
makespan)\footnote{By scaling the units of time and energy, we can assume without loss of
generality that $\rho = 1$.}. In fact, minimizing the sum of conflicting objectives has been a
common practice in many bi-criteria optimization problems \cite{AlbersFu07,LamLeTo08c}, and similar
metrics have been considered previously in the scheduling literature that combine both performance
and the cost of scheduling into a single objective function \cite{Trick90,ShmoysTa93,Chen04}.

Since Albers and Fujiwara \cite{AlbersFu07} first considered the problem of minimizing total flow
time plus energy, many results (e.g., \cite{BansalPrSt09, LamLeTo08, LamLeTo08c, BansalChPr09,
ChanEdLa11, ChanEdPr11, SunCaHs09b, GreinerNoSo09, AndrewWiTa09, AndrewLiWi10}) have been obtained
under different online scheduling settings. Some of these results assume that the scheduling
algorithm is \emph{clairvoyant}; that is, it gains complete knowledge of all job characteristics
immediately upon the job's arrival.  Other results are for an arguably more practical
\emph{non-clairvoyant} setting, where the scheduler knows nothing about the un-executed portion of
a job. Most of these results, however, are only applicable to scheduling sequential jobs.  Also, to
the best of our knowledge, no previous work has considered minimizing makespan plus energy. The
closest result to ours is by Chan, Edmonds and Pruhs \cite{ChanEdPr11}, who studied non-clairvoyant
scheduling for parallel jobs on multiprocessors to minimize total flow time plus energy. In both
\cite{ChanEdPr11} and our previous work \cite{SunCaHs09b}, it has been observed that any
non-clairvoyant algorithm that allocates a set of uniform-speed processors to a job will perform
poorly; in particular, a lower bound of $\Omega(P^{(\alpha-1)/\alpha^2})$ on the competitiveness
has been shown for any such algorithm, where $P$ is the total number of processors. The reason is
because a non-clairvoyant algorithm may in the worst case allocate a ``wrong" number of processors
to a job as compared to its parallelism, which will lead to either wasted energy or delayed job
execution.

To obtain a better competitive ratio, it turns out that a non-clairvoyant algorithm needs to assign
processors of different speeds to a job. To this end, Chan, Edmonds and Pruhs \cite{ChanEdPr11}
proposed an execution model, in which a job can be simultaneously executed by several groups of
processors. The processors within the same group must share the same speed, but different groups
can run at different speeds. The execution rate of the job at any time is determined by the group
with the fastest speed\footnote{In practice, this can be implemented by proper checkpointing of the
executing program.}. They proposed a non-clairvoyant algorithm called \MULTI, and showed that it is
$O(\log P)$-competitive with respect to total flow time plus energy for any set of parallel jobs.
They also gave an $\Omega(\log^{1/\alpha}P)$ lower bound on the competitive ratio of any
non-clairvoyant algorithm under this execution model.

In this paper, we first propose an alternative execution model, under which only one group of
processors, possibly with different speeds, can be allocated to a job at any time. The execution
rate of the job is determined by the speeds of the fastest processors that can be effectively
utilized. This model is based on the assumption that the \emph{maximum utilization policy}
\cite{Jaffe80a,BenderRa00} is employed at the underlying task scheduling level, which always
utilizes faster processors before slower ones. Compared to the execution model proposed in
\cite{ChanEdPr11}, our model may be implemented more easily especially for data-parallel jobs with
independent and sufficiently long tasks. Our first contribution includes a non-clairvoyant
scheduling algorithm and its analysis under this execution model. The following states our results:

\begin{itemize}
\item We propose a non-clairvoyant algorithm \NEQUI (Non-uniform Equi-partitioning), and show that it
  is $O(\ln P)$-competitive with respect to the total flow time plus
  energy for any set of parallel jobs with arbitrary release time,
  and $O(\ln^{1/\alpha}P)$-competitive for jobs with identical release time. Moreover, we prove that any non-clairvoyant algorithm is
  $\Omega(\ln^{1/\alpha}P)$-competitive under our execution model,
  showing that \NEQUI is asymptotically optimal in the batch-released
  setting.\footnote{It is interesting to observe that \NEQUI and \MULTI achieve the same asymptotic competitive ratio under two different execution models that are not clearly related.}
\end{itemize}

Another contribution of this paper is to study a setting that lies between clairvoyance and
non-clairvoyance. In this intermediate setting, a scheduling algorithm is allowed to know the
available parallelism, or the \emph{instantaneous parallelism (IP)}, of a job at any given time.
The future characteristic of the job, such as its remaining parallelism or work, is still unknown.
We call such an algorithm \emph{IP-clairvoyant}\footnote{This is to be distinguished from
\emph{semi-clairvoyant} scheduling \cite{BecchettiLeMa04}, which is another intermediate setting
that assumes a scheduling algorithm is able to gain approximate knowledge of a job upon its
arrival, such as an estimate of its total work, but not the job's exact information.}. In many
parallel systems using centralized task queues or thread pools, instantaneous parallelism is simply
the number of ready tasks in the queue or the number of ready threads in the pool, which is
information practically available to the scheduler. Even for parallel systems using distributed
scheduling such as work-stealing \cite{BlumofeLe99}, instantaneous parallelism can be collected or
estimated through counting or sampling without introducing much system overhead. It was shown
previously that, when minimizing total flow time alone, knowledge about the instantaneous
parallelism of the jobs provides limited benefit when compared to non-clairvoyant algorithms
\cite{DengGuBr00,EdmondsChBr03,KalyanasundaramPr00,EdmondsPr09}. However, we show in this paper
that IP-clairvoyance can bring significant performance improvements when it comes to minimizing
total flow time plus energy. Our contribution in this setting includes the following results:

\begin{itemize}
\item We present an IP-clairvoyant algorithm \UCEQ (Uniform Conservative Equi-partitioning), and show that it
  is $\paren{\max\{\frac{4\alpha^2}{\alpha-1}, 4^{\alpha}\alpha\} +
    2\alpha}$-competitive with respect to total flow time plus
  energy for any set of parallel jobs with arbitrary release
  time. This competitive ratio is independent of the total number $P$
  of processors, and therefore can be considered as constant for a
  fixed power parameter $\alpha$. In addition, we show that \UCEQ is
  $(2^{2-1/\alpha} + 2)$-competitive for any set of parallel jobs with identical release time.
\end{itemize}

\begin{table}[t]\label{tab:flow}
\caption{Competitive ratios of our non-clairvoyant and IP-clairvoyant algorithms for parallel jobs
with arbitrary and identical release time for total response time plus energy.} \centering
\begin{tabular}{| c | c | c |}
\hline
& Non-clairvoyant & IP-clairvoyant \\
\hline
Arbitrary release time & $O(\ln P)$ & $O(1)$ \\
\hline
Identical release time & $\Theta(\ln^{1/\alpha}P)$ & $2^{2-1/\alpha} + 2$ \\
\hline
\end{tabular}
\end{table}

\tabref{flow} summarizes the competitive ratios of our algorithms under both non-clairvoyant and
IP-clairvoyant settings. Compared to any non-clairvoyant algorithm, our IP-clairvoyant algorithm
achieves significantly better competitive ratios, and in particular it gives the first constant
competitive result on multiprocessor speed scaling for parallel jobs.  The reason for the
improvement comes from the fact that, given the instantaneous parallelism, an IP-clairvoyant
algorithm can now allocate a ``right" number of processors to a job at any time, ensuring that no
energy will be wasted. At the same time, it can also guarantee a sufficient execution rate by
setting the total power consumption proportionally to the number of active jobs at any time. This
has been a common practice to designing online scheduling algorithms for total flow time plus
energy and intuitively it provides the optimal balance between energy and performance
\cite{BansalPrSt09,LamLeTo08,BansalChPr09,ChanEdLa11}. Moreover, unlike the non-clairvoyant
algorithms \MULTI and \NEQUI, both of which require non-uniform speed scaling for an individual
job, \UCEQ only requires allocating processors of uniform speed to a job. Thus, in situations where
the instantaneous parallelism of a job does not change frequently and can be effectively measured,
e.g., by using feedback mechanisms \cite{AgrawalLeHe08,HeHsLe08,SunCaHs11}, our IP-clairvoyant
algorithm may be easier to implement and more practical.

Besides minimizing total flow time plus energy, there have been some recent studies that focus on
the weighted variant of this problem \cite{ChanLaLe10}, or optimize a linear combination of energy
and some other performance metrics, such as total profit \cite{PruhsSt10} and quoted lead time
\cite{ChanLaLi11}. In this paper, we introduce a new objective function of minimizing makespan plus
energy. Unlike the previous metrics, where the completion time of each job contributes to the
overall objective function, makespan is only determined by the completion time of the last job in a
job set, while the other jobs only contribute to the energy consumption part of the objective, and
therefore can be slowed down to improve the overall performance. However, without knowing the
future characteristics of the jobs, such as their remaining work, it is not clear even in the
IP-clairvoyant setting which jobs should be slowed down in order to reduce energy without affecting
the makespan. In the preliminary version \cite{SunHeHs11} of this paper, we proposed an
IP-clairvoyant algorithm that works for parallel jobs with identical release time and that consist
of sequential phases and fully parallelizable phases up to all $P$ processors. In this paper, we
develop a generalized strategy that works for any set of parallel jobs regardless of their release
time and parallelism structure.
The following shows our contribution for minimizing makespan plus energy:

\begin{itemize}
\item We present an IP-clairvoyant algorithm \PF (Work-Conserving Equal-Power) and show that it is
  $O(\ln^{1-1/\alpha}P)$-competitive with respect to makespan plus
  energy for any set of parallel jobs regardless of their release
  time, where $P$ is the total number of processors. Moreover, we
  give a matching $\Omega(\ln^{1-1/\alpha}P)$ lower bound on the
  competitive ratio of any IP-clairvoyant algorithm, showing that \PF is
  asymptotically optimal.
\end{itemize}

Finally, compared to minimizing total flow time plus energy, where a common strategy is to set the
total power consumption at any time proportionally to the number of active jobs
\cite{BansalPrSt09,BansalChPr09,ChanEdLa11,LamLeTo08}, our results indicate that a good strategy
for minimizing makespan plus energy is to set a constant power consumption at all time. In fact,
both strategies share the same principle of balancing the costs incurred from both the power
consumption and the target performance metric during the jobs' executions.

The rest of this paper is organized as follows. \secref{models} formally defines the models and the
objective functions. \secref{flow} presents our algorithms and analysis in both non-clairvoyant and
IP-clairvoyant settings for the objective of total flow time plus energy. \secref{makespan}
presents our IP-clairvoyant algorithm for minimizing makespan plus energy. Finally,
\secref{discussion} concludes the paper with some discussions and future directions.

\section{Models and Objective Functions}\label{sec:models}

We consider a set $\jobset =\{J_1, J_2, \cdots, J_n\}$ of $n$ jobs with time-varying parallelism to
be scheduled on $P$ processors whose speeds can be individually scaled. The power consumption of a
processor running at speed $s$ is given by $s^{\alpha}$, where $s$ can take any value in $[0,
\infty)$ and $\alpha > 1$ is the \emph{power parameter}. Adopting the notations used previously in
\cite{EdmondsChBr03,Edmonds00,EdmondsPr09,ChanEdPr11}, each job $J_i\in \jobset$ contains $k_i$
phases $\langle J_i^1, J_i^2, \cdots, J_i^{k_i} \rangle$, and each phase $J_i^{k}$ is represented
by an ordered pair $\langle w_i^{k}, h_i^{k}\rangle$, where $w_i^{k} \in \mathbb{R^+}$ denotes the
amount of \emph{work} and $h_i^{k} \in \mathbb{Z}^+$ denotes the \emph{parallelism} of the
phase\footnote{In \cite{EdmondsChBr03,Edmonds00,EdmondsPr09,ChanEdPr11}, an arbitrary
  non-decreasing and sub-linear speedup function is specified for each
  phase instead of a parallelism value, which represents a more
  general model for the jobs. For any non-clairvoyant algorithm,
  however, it was shown that the simple model used in this paper gives
  the hardest job instances for the combined objective of performance
  and energy \cite{ChanEdPr11}.}. Since a job can receive at most $P$ processors at any time, it
does not benefit by having a larger parallelism value than $P$. Hence, we can assume without loss
of generality that $h_i^{k} \le P$. A phase $J_i^{k}$ is said to be \emph{fully-parallelizable} if
$h_i^k = P$ and it is \emph{sequential} if $h_i^k = 1$. For convenience, we also define $x_i^k =
w_i^k/(h_i^k)^{1-1/\alpha}$ to be the \emph{unit-power span} for each phase $J_i^k$. This
represents the time to complete the phase using exactly $h_i^k$ processors of the same speed with a
total power of 1 at all time. Suppose the $h_i^k$ processors have the same speed $s$, we have
$h_i^k s^\alpha = 1$, and so $s = (h_i^k)^{-\frac{1}{\alpha}}$. The amount of time to complete
$w_i^k$ amount of work is thus given by $w_i^k / (s \cdot h_i^k) = x_i^k$. For each job $J_i$, let
$w(J_i) = \sum_{k=1}^{k_i}w_i^k$ denote its \emph{total work} and let $x(J_i) =
\sum_{k=1}^{k_i}x_i^k$ denote the job's \emph{total unit-power span}.

Suppose that at some time $t$ job $J_i$ is in its $k$'th phase hence has parallelism $h_i^k$, and
it is allocated $a_i(t)$ processors possibly with different speeds. Since a job cannot utilize more
processors than its parallelism, its \emph{effective processor allocation} at time $t$ is given by
$\bar{a}_i(t) = \min\{a_i(t), h_i^k\}$. The execution of the job is assumed to follow the
\emph{maximum utilization policy} \cite{Jaffe80a,BenderRa00}, which always utilizes faster
processors before slower ones until all the allocated processors are utilized or the number of
utilized processors reaches the parallelism of the job. In particular, let $s_{ij}(t)$ denote the
speed of the $j$'th processor allocated to job $J_i$ at time $t$, and we can assume without loss of
generality that $s_{i1}(t) \ge s_{i2}(t) \ge \cdots \ge s_{ia_i(t)}(t)$. Then, only the
$\bar{a}_i(t) = \min\{a_i(t), h_i^k\}$ fastest processors are utilized, and the \emph{execution
rate} $\Gamma_{i}^{k}(t)$ of the job is given by $\Gamma_{i}^{k}(t) =
\sum_{j=1}^{\bar{a}_i(t)}s_{ij}(t)$. In the case where all the processors allocated to job $J_i$
share the same speed $s_i(t)$, the execution rate is then simply $\Gamma_{i}^{k}(t) =
\bar{a}_i(t)s_{i}(t)$.

At any time $t$, a scheduling algorithm needs to specify the number $a_i(t)$ of processors
allocated to each job $J_i$, as well as the speed of each allocated processor. In this paper, we
study two types of algorithms. An algorithm is said to be \emph{non-clairvoyant} if it makes both
scheduling decisions without any current or future information about a job, such as its release
time, parallelism profile and remaining work. If an algorithm is aware of the current, or
\emph{instantaneous parallelism} of the job at any time but not its remaining work and parallelism,
the algorithm is said to be \emph{IP-clairvoyant}.

In any valid schedule, we require the total processor allocation at any time to be at most the
total number of available processors, i.e., $\sum_{i=1}^{n}a_i(t)\le P$. Let $r_i$ denote the
\emph{release time} of job $J_i$. If all jobs are released together, their release time can be
assumed to be all $0$. Otherwise, we can assume without loss of generality that the first released
job arrives at time $0$. Let $c_i$ denote the completion time of job $J_i$, and let $c_i^k$ denote
the completion time of phase $J_i^k$. We also require that a valid schedule cannot begin to execute
a phase of a job unless it has completed all its preceding phases, i.e., $r_i = c_{i}^{0} <
c_{i}^{1} < \cdots < c_{i}^{k_i} = c_i$, and $\int_{c_{i}^{k-1}}^{c_{i}^{k}}\Gamma_i^{k}(t) dt =
w_i^{k}$ for all $1 \le k \le k_i$.

The \emph{flow time} $f_{i}$ of any job $J_{i}$ is the duration between its completion and release,
i.e., $f_{i} = c_{i} - r_{i}$. The \emph{total flow time} $F(\jobset)$ of all jobs in $\jobset$ is
given by $F(\jobset) = \sum_{i=1}^{n}f_{i}$.  The \emph{makespan} $M(\jobset)$ is the completion
time of the last completed job, i.e., $M(\jobset) = \max_{i=1,\cdots, n}c_i$. Job $J_i$ is said to
be \emph{active} at time $t$ if it is released but not completed at $t$, i.e., $r_i \le t \le c_i$.
An alternative expression for the total flow time is $F(\jobset) = \int_{0}^{\infty}n_tdt$, where
$n_t$ is the number of active jobs at time $t$.
Let $u_i(t)$ denote the power consumed by job
  $J_i$ at time $t$, i.e., $u_i(t) =
  \sum_{j=1}^{a_i(t)}s_{ij}(t)^{\alpha}$. The overall energy
  consumption $e_i$ of the job is given by $e_i
  =\int_{0}^{\infty}u_i(t)dt$, and the total energy consumption
  $E(\jobset)$ of the job set is $E(\jobset) = \sum_{i=1}^{n}e_i$, or
  alternatively $E(\jobset) = \int_{0}^{\infty}u_tdt$, where $u_t =
  \sum_{i=1}^{n}u_i(t)$ denotes the total power consumption of all
  jobs at time $t$. In this paper, we consider total flow time plus
  energy $G(\jobset)$ and makespan plus energy $H(\jobset)$ of the job
  set, i.e., $G(\jobset) = F(\jobset) + E(\jobset)$ and $H(\jobset) =
  M(\jobset) + E(\jobset)$. The objective is to minimize either
  $G(\jobset)$ or $H(\jobset)$.

We use \emph{competitive analysis} \cite{BorodinEl98} to evaluate an online scheduling algorithm by
comparing its performance with that of an optimal offline scheduler. An online algorithm \A is said
to be \emph{$c_1$-competitive} with respect to total flow time plus energy if it satisfies
$G_{\A}(\jobset) \le c_1\cdot G_{\OPT}(\jobset)$ for any job set $\jobset$, where
$G_{\OPT}(\jobset)$ denotes the total flow time plus energy of $\jobset$ under an optimal offline
scheduler. Similarly, an online algorithm \B is said to be \emph{$c_2$-competitive} with respect to
makespan plus energy if for any job set $\jobset$ we have $H_{\B}(\jobset) \le c_2 \cdot
H_{\OPT}(\jobset)$, where $H_{\OPT}(\jobset)$ denotes the makespan plus energy of the job set under
an optimal offline scheduler.

\section{Total Flow Time Plus Energy}\label{sec:flow}

We consider the objective of total flow time plus energy in this section. We first present a
non-clairvoyant algorithm \NEQUI and analyze its performances for jobs with both arbitrary release
time and the same release time. We then derive a lower bound on the competitive ratio of any
non-clairvoyant algorithm. Finally, we present an IP-clairvoyant algorithm \UCEQ and show that it
significantly improves upon any non-clairvoyant algorithm.

\subsection{Preliminaries}

We first derive two lower bounds on the total flow time plus energy of any job set, which allows us
to bound the performance of our online algorithms through indirect comparisons instead of comparing
directly to the optimal offline scheduler. We then introduce some useful notations, and outline the
analysis techniques used to prove the competitiveness of our online algorithms.

\subsubsection{Lower Bounds on Total Flow Time plus Energy}

Without loss of generality, we assume that the jobs in a job set $\jobset$ are renamed in
non-increasing order of total work, i.e., $w(J_1) \ge w(J_2) \ge \cdots \ge w(J_n)$. The following
lemma gives two lower bounds $G_1^*(\jobset)$ and $G_2^*(\jobset)$ on the total flow time plus
energy of job set $\jobset$. Note that the second lower bound only applies to a set of jobs with
identical release time; an algorithm may incur a smaller total flow time plus energy than
$G_2^*(\jobset)$ if the jobs in $\jobset$ have arbitrary release time.

\begin{lemma}\label{lem:lowerbounds}
The total flow time plus energy of any job set $\jobset$ satisfies the following two lower bounds,
i.e., $G_{\OPT}(\jobset) \ge \max\{G_1^*(\jobset), G_2^*(\jobset)\}$,
\begin{eqnarray}
G_1^*(\jobset) &=& \frac{\alpha}{(\alpha-1)^{1-1/\alpha}}\sum_{i=1}^{n}x(J_i), \\
G_2^*(\jobset) &=& \frac{\alpha}{\paren{(\alpha-1)P}^{1-1/\alpha}}\sum_{i=1}^{n}i^{1-1/\alpha}\cdot
w(J_i),
\end{eqnarray}
where $x(J_i)$ and $w(J_i)$ denote the unit-power span and the total work of job $J_i$,
respectively, and $P$ is the total number of processors. The second lower bound $G_2^*(\jobset)$
only applies to jobs with identical release time.
\end{lemma}

\begin{proof}
To derive the first lower bound, consider any phase $J_i^k$ of job $J_i$. The optimal scheduler
will only perform better if there is an unlimited number of processors at its disposal. In this
case, it will allocate $a$ processors of the same speed, say $s$, to the phase throughout its
execution, since the convexity of the power function implies that if different speeds are used,
then averaging the speeds will result in the same execution rate but consuming less energy
\cite{YaoDeSh95}. Moreover, we have $a \le h_i^k$, since allocating more processors to a phase than
its parallelism will incur more energy without improving flow time. The flow time plus energy
introduced by the execution of $J_i^k$ is then given by $\frac{w_i^k}{as} + \frac{w_i^k}{as}\cdot
as^{\alpha} = w_i^k\paren{\frac{1}{as} +
  s^{\alpha-1}} \ge \frac{\alpha}{(\alpha-1)^{1-1/\alpha}}\cdot
\frac{w_i^k}{a^{1-1/\alpha}} \ge \frac{\alpha}{(\alpha-1)^{1-1/\alpha}}\cdot
\frac{w_i^k}{(h_i^k)^{1-1/\alpha}} = \frac{\alpha}{(\alpha-1)^{1-1/\alpha}}\cdot x_i^k$. Extending
this property over all phases and all jobs gives the first lower bound.

For the second lower bound, the optimal offline scheduler will do no worse in terms of total flow
time plus energy if each job in the job set is replaced by a simpler job that contains a single
fully-parallelizable phase with work $w_i$, since the original optimal schedule is also a valid
schedule for the new job set. Also, because the jobs in the job set are assumed to have the same
release time, it is well-known that the optimal offline scheduler will execute them using the \SJF
(Shortest Job First) policy, since otherwise the total flow time can be reduced by swapping the
jobs without affecting the energy consumption. Moreover, for each job $J_i$, the optimal offline
scheduler will allocate all $P$ processors the same speed, say $s$, throughout its execution, by
the same argument as in the proof of the first lower bound. The flow time plus energy introduced by
the execution of $J_i$ is then given by $\frac{w(J_i)}{Ps}\cdot i + \frac{w(J_i)}{Ps}\cdot
Ps^{\alpha} = w(J_i)\paren{\frac{i}{Ps} +
  s^{\alpha-1}} \ge \frac{\alpha}{(\alpha-1)^{1-1/\alpha}}\cdot
\frac{i^{1-1/\alpha}\cdot w(J_i)}{P^{1-1/\alpha}}$. Summing the inequality over all jobs gives the
second lower bound.
\end{proof}

\subsubsection{Concepts and Notations}\label{sec:notations}

We define some useful concepts and notations in this subsection in order to analyze the performance
of any online algorithm \A.

First, let $\jobset(t)$ and $\jobset^*(t)$ denote the sets of active jobs at any time $t$ scheduled
by an online algorithm \A and the optimal offline algorithm, respectively. Since an online
algorithm and the optimal algorithm may schedule the same job set differently, $\jobset(t)$ and
$\jobset^*(t)$ can be different from each other at any time instance. We define
$\frac{dG_{\A}(\jobset(t))}{dt}$ to be the \emph{instantaneous cost} of the online algorithm \A
with respect to the total flow time plus energy at time $t$, and define
$\frac{dG_{\OPT}(\jobset^*(t))}{dt}$ to be the instantaneous cost of the optimal offline algorithm
at time $t$. Since the online algorithm contains $n_t$ active jobs and consumes $u_t$ power at time
$t$, its instantaneous cost is given by $\frac{dG_{\A}(\jobset(t))}{dt} = n_t + u_t$, and its total
flow time plus energy for the entire job set can be obtained by integrating the above instantaneous
cost over time, i.e., $G_A(\jobset) = \int_{0}^{\infty}\frac{dG_{\A}(\jobset(t))}{dt} dt =
\int_{0}^{\infty} (n_t + u_t) dt$. Similarly, we have $\frac{dG_{\OPT}(\jobset^*(t))}{dt} = n_t^* +
u_t^*$ for the optimal offline algorithm, where $n_t^*$ and $u_t^*$ denote the number of active
jobs and the power consumption under the optimal algorithm at time $t$. The total flow time plus
energy incurred by the optimal offline algorithm is then given by $G_{\OPT}(\jobset) =
\int_{0}^{\infty}\frac{dG_{\OPT}(\jobset^*(t))}{dt} dt = \int_{0}^{\infty} (n^*_t + u^*_t) dt$.

Now, let us define the notions of $t$-prefix and $t$-suffix. Specifically, the \emph{t-prefix}
$J_i(\la{t})$ of a job $J_i$ is defined as the portion of the job scheduled by the online algorithm
$\A$ no later than time $t$, and the \emph{t-suffix} $J_i(\ra{t})$ is the portion of job $J_i$
scheduled by algorithm $\A$ after time $t$. Moreover, we can extend the notions of $t$-prefix and
$t$-suffix from an individual job to a job set as follows. The $t$-prefix of a job set $\jobset$
scheduled by the online algorithm $\A$ is defined as $\jobset(\la{t}) = \{J_{i}(\la{t}): J_{i}\in
\jobset \mbox{ and } r_{i}\le t\}$ and the $t$-suffix of job set $\jobset$ is defined as
$\jobset(\ra{t}) = \{J_{i}(\ra{t}): J_{i}\in \jobset \mbox{ and } r_{i}\le t\}$.
With the help of these notions, we define $\frac{dG_1^*(\jobset(t))}{dt} =
\frac{G_1^*(\jobset(\la{t+\dt})) - G_1^*(\jobset(\la{t}))}{\dt}$ and $\frac{dG_2^*(\jobset(t))}{dt}
= \frac{G_2^*(\jobset(\la{t+\dt})) - G_2^*(\jobset(\la{t}))}{\dt}$ to be the rates of change at any
time $t$ for the two lower bounds presented in \lemref{lowerbounds}, where $\dt$ represents an
infinitesimally small interval of time during which no job arrives or completes. Note that these
two rates of change are defined with respect to the $t$-prefix, or the completed portion, of the
job set $\jobset$ scheduled by the online algorithm $\A$. From the definitions of the two lower
bounds, we can see that $\frac{dG_1^*(\jobset(t))}{dt}$ and $\frac{dG_2^*(\jobset(t))}{dt}$ are
essentially determined by how much work is done and how much unit-power span is completed for the
jobs at time $t$ under the online algorithm \A. For instance, if algorithm \A does not schedule any
job at time $t$, the $t$-prefix of the job set will not change, i.e., $\jobset(\la{t+\dt}) =
\jobset(\la{t})$, and as a result we will have $\frac{dG_1^*(\jobset(t))}{dt} =
\frac{dG_2^*(\jobset(t))}{dt} = 0$. If an online algorithm always schedules some job at all times,
the following lemma relates the two rates of change to the two lower bounds.

\begin{lemma}\label{lem:rateofchange}
For any job set $\jobset$ scheduled by an online algorithm \A, which always schedules some job at
all times, it satisfies that $G_1^*(\jobset) = \int_{0}^{\infty} \frac{dG_1^*(\jobset(t))}{dt} dt$
and $G_2^*(\jobset) = \int_{0}^{\infty} \frac{dG_2^*(\jobset(t))}{dt} dt$.
\end{lemma}

\begin{proof}
According to definition, the $t$-prefix $J_i(\la{t})$ of a job $J_i$ does not belong to
$\jobset(\la{t})$ if the job is not yet released at time $t$, but the job will remain in
$\jobset(\la{t})$ even after it has been completed. This ensures that the completion of a job will
not decrease $\frac{dG_1^*(\jobset(t))}{dt}$ and $\frac{dG_2^*(\jobset(t))}{dt}$, so they are
non-negative at all times. Since the online algorithm \A always schedules some active job, all jobs
are guaranteed to complete in finite amount of time. Thus, we can express the two lower bounds by
integrating their rates of change over time, i.e., $G_1^*(\jobset) = \int_{0}^{\infty}
\frac{dG_1^*(\jobset(t))}{dt} dt$ and $G_2^*(\jobset) = \int_{0}^{\infty}
\frac{dG_2^*(\jobset(t))}{dt} dt$.
\end{proof}

Finally, for analyzing the performance of an online algorithm \A for a set of jobs with arbitrary
release time, we also need to define a potential function $\Phi(t)$, whose form is usually
associated with the status of the job set at any time $t$ under both online algorithm and the
optimal offline algorithm \cite{ImMoPr11}. We can similarly define $\frac{d\Phi(t)}{dt} =
\frac{\Phi(t+\dt) - \Phi(t)}{\dt}$ to be the rate of change for the potential function at time $t$.

\subsubsection{Analysis Techniques}

We now outline two analysis techniques for proving the competitiveness of any online scheduling
algorithm. They are commonly known as the amortized local competitiveness argument and the local
competitiveness argument in the literature \cite{Pruhs07, ImMoPr11}. Both techniques compare the
cost of an online algorithm at any local time instance, or its instantaneous cost, with respect to
that of an optimal offline scheduler. For arbitrarily released jobs, the comparison is performed
with the help of the first lower bound given in \lemref{lowerbounds} and a carefully designed
potential function. For jobs with identical release time, both lower bounds are used to represent
the performance of the optimal.

The following lemma first illustrates the use of \emph{amortized local competitiveness argument}
for jobs with arbitrary release time. The technique arrives at the competitive ratio of any online
algorithm \A by bounding its instantaneous cost at any time $t$ with respect to the optimal offline
scheduler.

\begin{lemma}\label{lem:amortized}
Suppose that an online algorithm \A schedules a set $\jobset$ of jobs with arbitrary release time.
Then \A is $(c_1 + c_2)$-competitive with respect to total flow time plus energy, if given a
potential function $\Phi(t)$, the execution of the job set satisfies the following

- Boundary condition: $\Phi(0)\le 0$ and $\Phi(\infty)\ge 0$;

- Arrival condition: $\Phi(t)$ does not increase whenever a new job arrives;

- Completion condition: $\Phi(t)$ does not increase whenever a job completes under either \A or the
optimal offline scheduler;

- Running condition: $\frac{dG_{\A}(\jobset(t))}{dt} + \frac{d\Phi(t)}{dt} \le c_1 \cdot
\frac{dG_{\OPT}(\jobset^*(t))}{dt} + c_2 \cdot \frac{dG_1^*(\jobset(t))}{dt}$ at all time $t \ge
0$.
\end{lemma}
\begin{proof}
Let $T$ denote the set of time instances when a job arrives or completes under either the online
algorithm \A or the optimal offline scheduler. Integrating the running condition over time and
applying \lemref{rateofchange}, we get $G_{\A}(\jobset) + \Phi(\infty) - \Phi(0) + \sum_{t\in
  T}\paren{\Phi(t^-) - \Phi(t^+)} \le c_1 \cdot G_{\OPT}(\jobset) + c_2
\cdot G_1^*(\jobset)$, where $t^-$ and $t^+$ denote the times right before and after the event
occurred at time $t$. Now, applying boundary, arrival and completion conditions to the above
inequality, we get $G_{\A}(\jobset) \le c_1 \cdot G_{\OPT}(\jobset) + c_2 \cdot G_1^*(\jobset)$.
Since $G_1^*(\jobset)$ is a lower bound on the total flow time plus energy of job set $\jobset$
according to \lemref{lowerbounds}, the performance of algorithm \A satisfies $G_{\A}(\jobset) \le
(c_1 + c_2) \cdot G_{\OPT}(\jobset)$.
\end{proof}

When scheduling for a set of jobs with identical release time, the analysis turns out to be simpler
as the potential function is usually not needed. In this case, we can get the competitive ratio of
online algorithm \A by using the \emph{local competitiveness argument}, which directly compares its
instantaneous cost at any time $t$ with respect to the rates of change for both lower bounds given
in \lemref{lowerbounds}. The following lemma illustrates this technique.

\begin{lemma}\label{lem:local}
Suppose that an online algorithm \A schedules a set $\jobset$ of jobs with identical release time.
Then \A is $(c_1 + c_2)$-competitive with respect to total flow time plus energy, if the execution
of the job set satisfies the following

- Running condition: $\frac{dG_{\A}(\jobset(t))}{dt} \le c_1 \cdot \frac{dG_1^*(\jobset(t))}{dt} +
c_2 \cdot \frac{dG_2^*(\jobset(t))}{dt}$ at all time $t\ge 0$.
\end{lemma}

\begin{proof}
Similarly to the proof of \lemref{amortized}, by integrating the running condition over time and
applying \lemref{rateofchange}, we get $G_{\A}(\jobset) \le c_1 \cdot G_1^*(\jobset) + c_2 \cdot
G_2^*(\jobset)$. Since both $G_1^*(\jobset)$ and $G_2^*(\jobset)$ are lower bounds on the total
flow time plus energy of job set $\jobset$ according to \lemref{lowerbounds}, the result follows.
\end{proof}

\subsection{Non-clairvoyant Algorithm: \NEQUI}\label{sec:nequi}

It was shown in \cite{ChanEdPr11,SunCaHs09b} that any non-clairvoyant algorithm which allocates a
set of uniform-speed processors to a job is $\Omega(P^{(\alpha-1)/\alpha^2})$-competitive, where
$P$ is the total number of processors. To achieve better performance, we propose a non-clairvoyant
algorithm called \NEQUI (Non-uniform Equi-partitioning), which equally partitions the $P$
processors among the $n_t$ active jobs at any time $t$. \algref{nequi} describes its details.

Specifically, when the number of processors is at least the number of active jobs, i.e., $P \ge
n_t$, it sets the speeds of the allocated processors for each active job in a non-uniform manner.
Intuitively, since the algorithm does not know the parallelism of a job to guide its processor
allocation, the non-uniform speed assignment balances the waste of energy due to the possible
overallocation and the delay of the job due to the possible underallocation. On the other hand,
when $n_t
> P$, it assigns the same speed to all processors and relies on time-sharing to allocate $P/n_t$
fraction of a processor to each active job, which is commonly implemented in the operating system
using the round robin policy. Note that since the parallelism of any active job is at least 1, no
energy waste will be incurred in this case.


\begin{algorithm}[h]
\caption{\NEQUI}\label{alg:nequi}
\begin{algorithmic}[1]
\REQUIRE{total number $P$ of processors and number $n_t$ of active jobs at time $t$. }
\ENSURE{number of allocated processors and their speeds for each active job at time $t$.}

\IF {$P \ge n_t$} \STATE{allocate $a_i(t) = \floor{\frac{P}{n_t}}$ processors to each active job
$J_i$.} \STATE {set the speed of the $j$'th processor allocated to job $J_i$ to be $s_{ij}(t) =
\paren{\frac{1}{(\alpha-1)H_P \cdot j}}^{1/\alpha}$, where
$j=1, \cdots, a_i(t)$ and $H_P = \sum_{k=1}^{P}\frac{1}{k}$ is the $P$'th harmonic number.} \ELSE
\STATE{allocate $a_i(t) = \frac{P}{n_t}$ fraction of a processor to each active job $J_i$.}
\STATE{set the speed of all processors to be $s(t) =
\paren{\frac{n_t}{(\alpha-1)H_P \cdot P}}^{1/\alpha}$.} \ENDIF
\end{algorithmic}
\end{algorithm}

At any time $t$ when job $J_i$ is in its $k$'th phase, we say that it is \emph{satisfied} if its
processor allocation is at least its instantaneous parallelism, i.e., $a_i(t) \ge h_i^k$.
Otherwise, the job is said to be \emph{deprived}. Let $\jobset(t)$ denote the set of all active
jobs at time $t$, and let $\jobset_S(t)$ and $\jobset_D(t)$ denote the set of satisfied and
deprived jobs at time $t$, respectively. For convenience, we let $n_t^S = \card{\jobset_S(t)}$ and
$n_t^D = \card{\jobset_D(t)}$. Since an active job is either satisfied or deprived, we have
$\card{\jobset(t)} = n_t = n_t^S + n_t^D$. Moreover, we define $x_t = n_t^D / n_t$ to be the
\emph{deprived ratio} at time $t$. To assist analysis, we first bound the execution rate and the
power consumption of \NEQUI for any active job at time $t$ in the following lemma.

\begin{lemma}\label{lem:ratebounds}
Suppose that \NEQUI schedules a set $\jobset$ of jobs. Then for any job $J_i \in \jobset$, its
execution rate at time $t$ satisfies
$\paren{\frac{1}{(\alpha-1)H_P}}^{1/\alpha}\frac{\bar{a}_i(t)^{1-1/\alpha}}{2^{1/\alpha}} \le
\Gamma_i^k(t) \le\paren{\frac{1}{(\alpha-1)H_P}}^{1/\alpha}
\frac{\bar{a}_i(t)^{1-1/\alpha}}{1-1/\alpha}$, where $\bar{a}_i(t) = \min\{a_i(t), h_i^k\}$ denotes
the effective processor allocation for job $J_i$ at time $t$. Also, the power consumption of the
job at time $t$ satisfies $u_i(t)\le\frac{1}{\alpha-1}$.
\end{lemma}

\begin{proof}
When $P < n_t$, we have $a_i(t) = P/n_t < 1 \le h_i^k$ for job $J_i$, so $\bar{a}_i(t) = a_i(t)$.
The execution rate of the job is given by $\Gamma_i^k(t) = a_i(t)s(t) =
\paren{\frac{1}{(\alpha-1)H_P}}^{1/\alpha}\bar{a}_i(t)^{1-1/\alpha}$, and its power consumed is
$u_i(t) = \frac{1}{(\alpha-1)H_P} \le \frac{1}{\alpha-1}$.

When $P \ge n_t$, we have $\bar{a}_i(t) \ge 1$, since $a_i(t) = \floor{P/n_t} \ge 1$ and $h_i^k \ge
1$. The execution rate of the job is $\Gamma_i^k(t) = \sum_{j=1}^{\bar{a}_i(t)}s_{ij}(t) =
\paren{\frac{1}{(\alpha-1)H_P}}^{1/\alpha}\sum_{j=1}^{\bar{a}_i(t)}\frac{1}{j^{1/\alpha}}
$, which can be approximated with integration: $\sum_{j=1}^{\bar{a}_i(t)}s_{ij}(t) \le
\paren{\frac{1}{(\alpha-1)H_P}}^{1/\alpha}
\int_{0}^{\bar{a}_i(t)}\frac{1}{j^{1/\alpha}}dj =
\paren{\frac{1}{(\alpha-1)H_P}}^{1/\alpha}
\frac{\bar{a}_i(t)^{1-1/\alpha}}{1-1/\alpha}$, and $\sum_{j=1}^{\bar{a}_i(t)}s_{ij}(t) \ge
\paren{\frac{1}{(\alpha-1)H_P}}^{1/\alpha}
\int_{1}^{\bar{a}_i(t)+1}\frac{1}{j^{1/\alpha}}dj =
\paren{\frac{1}{(\alpha-1)H_P}}^{1/\alpha}
\frac{(\bar{a}_i(t)+1)^{1-1/\alpha}-1}{1-1/\alpha} \ge
\paren{\frac{1}{(\alpha-1)H_P}}^{1/\alpha} \cdot
\frac{2^{1-1/\alpha}-1}{1-1/\alpha}\bar{a}_i(t)^{1-1/\alpha} \ge
\paren{\frac{1}{(\alpha-1)H_P}}^{1/\alpha}
\frac{\bar{a}_i(t)^{1-1/\alpha}}{2^{1/\alpha}}$. The second to last inequality follows because
$\frac{(x+1)^{1-1/\alpha}-1}{x^{1-1/\alpha}}$ is an increasing function of $x$ for all $x > 0$ and
we have $\bar{a}_i(t) \ge 1$. The power consumption of the job satisfies $u_i(t) =
\sum_{j=1}^{a_i(t)}s_{ij}(t)^\alpha = \frac{1}{(\alpha-1)H_P}\sum_{j=1}^{a_i(t)}\frac{1}{j} =
\frac{H_{a_i(t)}}{(\alpha-1)H_P} \le \frac{1}{\alpha-1}$, where $H_{a_i(t)}$ is $a_i(t)$'th
harmonic number, and $H_{a_i(t)} \le H_P$ since $a_i(t) \le P $ at all times.
\end{proof}

\subsubsection{Performance for Jobs with Arbitrary Release Time}

We first bound the performance of \NEQUI for a set of jobs with arbitrary release time. We adopt
the potential function proposed by Lam et al.~\cite{LamLeTo08} in the analysis of an online speed
scaling algorithm for sequential jobs. Specifically, we focus on the $t$-suffix $\jobset(\ra{t})$
of job set $\jobset$ and define $n_t(z)$ to be the number of active jobs whose remaining work is at
least $z$ at time $t$ under \NEQUI, i.e., $n_t(z) = \sum_{J_i \in \jobset(t)}[w(J_i(\ra{t})) \ge
z]$, where $[x]$ is $1$ if proposition $x$ is true and $0$ otherwise. Also, define $n_t^*(z)$ to be
the number of active jobs whose remaining work is at least $z$ at time $t$ under the optimal
offline algorithm. The potential function is then defined as
\begin{equation}\label{eq:potential}
\Phi(t) = \eta \int_{0}^{\infty}\left[\paren{\sum_{i=1}^{n_t(z)}i^{1-1/\alpha}} -
n_t(z)^{1-1/\alpha}n_t^*(z)\right]dz,
\end{equation}
where $\eta = \eta' \frac{H_P^{1/\alpha}}{P^{1-1/\alpha}}$ and $\eta'$ is a constant to be
specified later. In particular, the integration of the first term of the potential function is
proportional to the optimal total flow time plus energy for the remaining portion, or $t$-suffix,
of the job set at time $t$ scheduled under \NEQUI. The second term of the potential function is
added to ensure that the arrival condition will be satisfied. (See proof of \thmref{release}.)

In addition, we need to use the following lemma in our proof.

\begin{lemma}\label{lem:young}
For any $n_t \ge 0$, $s_j^*\ge 0$ and $\lambda > 0$, we have that $n_t^{1-1/\alpha}s_j^{*} \le
\frac{\lambda
\paren{H_P \cdot P}^{1-1/\alpha}}{\alpha}\paren{s_j^{*}}^{\alpha} +
\frac{1-1/\alpha}{\lambda^{1/(\alpha-1)}\paren{H_P \cdot P}^{1/\alpha}}n_t$.
\end{lemma}

\begin{proof}
The lemma is a direct result of Young's Inequality \cite{HardyLiPo52}, which is stated formally as
follows. If $f$ is a continuous and strictly increasing function on $[0, c]$ with $c > 0$, $f(0) =
0$, $a \in [0, c]$ and $b \in [0, f(c)]$, then $ab \le \int_{0}^{a}f(x)dx +
\int_{0}^{b}f^{-1}(x)dx$, where $f^{-1}$ is the inverse function of $f$. By setting $f(x) = \lambda
\paren{H_P\cdot P}^{1-1/\alpha} x^{\alpha - 1}$, $a = s_j^*$ and $b = n_t^{1-1/\alpha}$, the lemma
is directly implied.
\end{proof}

Using \lemrefthree{amortized}{ratebounds}{young}, we now prove the competitive ratio of \NEQUI for
jobs with arbitrary release time.

\begin{theorem}\label{thm:release}
\NEQUI is $O(\ln P)$-competitive with respect to total flow time plus energy for any set of
parallel jobs, where $P$ is the total number of processors.
\end{theorem}

\begin{proof}
We will show that the execution of any job set scheduled by \NEQUI (\NE for short) satisfies the
boundary, arrival and completion conditions in \lemref{amortized}, as well as the running condition
$\frac{dG_{\NE}(\jobset(t))}{dt} + \frac{d\Phi(t)}{dt} \le c_1 \cdot
\frac{dG_{\OPT}(\jobset^*(t))}{dt} + c_2 \cdot \frac{dG_1^*(\jobset(t))}{dt}$, where $c_1 = O(\ln
P)$ and $c_2 = O(\ln^{1/\alpha}P)$. The theorem then follows by \lemref{amortized}.

- \emph{Boundary condition}: At time $0$, no job exists, so we have $n_t(z)=n_t^*(z)=0$ for all
$z\ge 0$, and so $\Phi(0) = 0$. At time $\infty$, all jobs are completed, so again we have
$\Phi(\infty) = 0$. Hence, the boundary condition is satisfied.

- \emph{Arrival condition}: Let $t^-$ and $t^+$ denote the time instances right before and after a
new job with work $w$ arrives at time $t$. Then we have $n_{t^+}(z) = n_{t^-}(z) + 1$ for $z \le w$
and $n_{t^+}(z) = n_{t^-}(z)$ for $z > w$, and similarly $n_{t^+}^*(z) = n_{t^-}^*(z) + 1$ for $z
\le w$ and $n_{t^+}^*(z) = n_{t^-}^*(z)$ for $z > w$. For convenience, we define $\phi_t(z) =
\paren{\sum_{i=1}^{n_t(z)}i^{1-1/\alpha}} -
n_t(z)^{1-1/\alpha}n_t^*(z)$. It is obvious that for $z > w$, we have $\phi_{t^+}(z) =
\phi_{t^-}(z)$. For $z \le w$, we can get $\phi_{t^+}(z) - \phi_{t^-}(z) =
n_{t^-}^*(z)\paren{n_{t^{-}}(z)^{1-1/\alpha} -
  (n_{t^{-}}(z)+1)^{1-1/\alpha}} \le 0$. Hence, $\Phi(t^+) = \eta
\int_{0}^{\infty} \phi_{t^+}(z)dz \le \eta \int_{0}^{\infty} \phi_{t^-}(z)dz = \Phi(t^-)$, and the
arrival condition is satisfied.

- \emph{Completion condition}: When a job completes under either \NEQUI or the optimal schedule,
$\Phi(t)$ is unchanged because $n_t(z)$ or $n^*_t(z)$ reduces by 1 for $z=0$ but does not change
for all $z > 0$. Therefore, the completion condition is satisfied.

- \emph{Running condition}: According to \lemref{ratebounds}, the overall power consumption $u_t$
of all active jobs at time $t$ under \NEQUI satisfies $u_t \le \frac{n_t}{\alpha-1}$. Thus, we have
$\frac{dG_{\NE}(\jobset(t))}{dt} = n_t + u_t \le \frac{\alpha}{\alpha-1}n_t$. Suppose the optimal
offline scheduler sets the speed of the $j$'th processor to $s_j^*$ at time $t$, which then gives
$\frac{dG_{\OPT}(\jobset^*(t))}{dt} = n_t^* + u_t^* = n_t^* +
\sum_{j=1}^{P}\paren{s_j^*}^{\alpha}$. To bound the rate of change $\frac{dG_1^*(\jobset(t))}{dt}$,
it turns out to be sufficient to consider the set $J_S(t)$ of satisfied jobs. Specifically, for
each satisfied job $J_i \in J_S(t)$, if it is in its $k$'th phase at time $t$ under \NEQUI, then we
have $a_i(t) \ge h_i^k$. According to \lemref{ratebounds}, the execution rate of the job is given
by $\Gamma_i^{k}(t) \ge
\paren{\frac{1}{(\alpha-1)H_P}}^{1/\alpha} \frac{\paren{h_i^k}^{1-1/\alpha}}{2^{1/\alpha}}$. Since
$\frac{dG_1^*(\jobset(t))}{dt}$ only depends on the unit-power span completed for the jobs at time
$t$ under the schedule of \NEQUI, we have $\frac{dG_1^*(\jobset(t))}{dt} \ge
\frac{\alpha}{(\alpha-1)^{1-1/\alpha}} \sum_{J_i\in
  \jobset_S(t)}\frac{\Gamma_i^{k}(t)}{\paren{h_i^k}^{1-1/\alpha}} \ge
\frac{\alpha}{\alpha-1}\paren{\frac{1}{2H_P}}^{1/\alpha}n_t^S =
\frac{\alpha}{\alpha-1}\paren{\frac{1}{2H_P}}^{1/\alpha}(1-x_t)n_t$, where $x_t$ is the deprived
ratio.

Now, we focus on finding an upper bound for the rate of change $\frac{d\Phi(t)}{dt}$ of the
potential function $\Phi(t)$ at time $t$. To this end, we mainly consider the set $\jobset_D(t)$ of
deprived jobs. If a deprived job $J_i \in \jobset_D(t)$ is in its $k$'th phase at time $t$ under
\NEQUI, then $a_i(t) < h_i^k$, and the execution rate of the job satisfies
$\paren{\frac{1}{(\alpha-1)H_P}}^{1/\alpha}\frac{a_i(t)^{1-1/\alpha}}{2^{1/\alpha}} \le
\Gamma_i^k(t) \le\paren{\frac{1}{(\alpha-1)H_P}}^{1/\alpha} \frac{a_i(t)^{1-1/\alpha}}{1-1/\alpha}$
by \lemref{ratebounds}. In the worst case, the $n_t^D$ deprived jobs may have the most remaining
work, so they will take the smallest $n_t^D$ indices in the first term of the potential function
given in \eqref{potential}. While considering the set of satisfied jobs can further decrease
$\Phi(t)$, we ignore these jobs for deriving an upper bound on the rate of change for the first
term of the potential function, and use all active jobs for bounding the rate of change for the
second term. The overall rate of change for the potential function can then bounded by
\begin{eqnarray*}\label{ineq:phi}
\frac{d\Phi(t)}{dt} &=& \eta \cdot
\frac{d}{dt}\int_{0}^{\infty}\left[\paren{\sum_{i=1}^{n_t(z)}i^{1-1/\alpha}} -
n_t(z)^{1-1/\alpha}n_t^*(z)\right]dz \nonumber \\
&\le& \frac{\eta}{\dt}\int_{0}^{\infty}\left[\paren{\sum_{i=1}^{n_{t+\dt}(z)} i^{1-1/\alpha}} -
\paren{\sum_{i=1}^{n_{t}(z)} i^{1-1/\alpha}}\right]dz \nonumber \\
    &&+ \frac{\eta}{\dt}\int_{0}^{\infty}\left[n_t(z)^{1-1/\alpha}\paren{n_t^*(z)
    - n_{t+\dt}^{*}(z)} + n_{t}^*(z)\paren{n_t(z)^{1-1/\alpha} - n_{t+\dt}(z)^{1-1/\alpha}}\right]dz \nonumber \\
&\le& \frac{\eta'H_P^{1/\alpha}}{P^{1-1/\alpha}}
\paren{-\sum_{i=1}^{n_t^D}i^{1-1/\alpha}\cdot \Gamma_i^{k}(t) +
  n_t^{1-1/\alpha} \sum_{j=1}^{P}s_j^* + n_t^*
  \sum_{i=1}^{n_t}\paren{i^{1-1/\alpha} - (i-1)^{1-1/\alpha}}
  \Gamma_i^{k}(t)}.
\end{eqnarray*}

We can get $\sum_{i=1}^{n_t^D}i^{1-1/\alpha} \ge \int_{0}^{n_t^D}i^{1-1/\alpha}di =
\frac{\paren{n_t^D}^{2-1/\alpha}}{2-1/\alpha} \ge \frac{x_t^2n_t^{2-1/\alpha}}{2}$ and
$\sum_{i=1}^{n_t}\paren{i^{1-1/\alpha} - (i-1)^{1-1/\alpha}} = n_t^{1-1/\alpha}$. Moreover,
according to \lemref{young}, we have $n_t^{1-1/\alpha}\sum_{j=1}^{P}s_j^* \le \frac{\lambda
\paren{H_P \cdot P}^{1-1/\alpha}}{\alpha}\sum_{j=1}^{P}\paren{s_j^*}^{\alpha}
+ \frac{1-1/\alpha}{\lambda^{1/(\alpha-1)}\paren{H_P \cdot P}^{1/\alpha}} P n_t$, where $\lambda$
is a constant to be specified later. For each satisfied job $J_i\in J_S(t)$, we have $h_i^k\le
a_i(t)$, so its execution rate satisfies $\Gamma_i^k(t) \le
\paren{\frac{1}{(\alpha-1)H_P}}^{1/\alpha} \frac{\paren{h_i^k}^{1-1/\alpha}}{1-1/\alpha} \le
\paren{\frac{1}{(\alpha-1)H_P}}^{1/\alpha} \frac{a_i(t)^{1-1/\alpha}}{1-1/\alpha}$ by \lemref{ratebounds}.
Finally, for each job $J_i \in \jobset(t)$, we have $\frac{P}{2n_t} \le \floor{\frac{P}{n_t}} \le
a_i(t) \le \frac{P}{n_t}$. Substituting these bounds into $\frac{d\Phi(t)}{dt}$ above and
simplifying, we have
\begin{equation}\label{ineq:phi}
\frac{d\Phi(t)}{dt} \le \eta' \paren{-\frac{x_t^2}{4(\alpha-1)^{1/\alpha}}n_t + \frac{\lambda
H_P}{\alpha}\sum_{j=1}^{P}\paren{s_j^*}^{\alpha} + \frac{1-1/\alpha}{\lambda^{1/(\alpha-1)}}n_t +
\frac{\alpha}{(\alpha-1)^{1+1/\alpha}}n_t^*}.
\end{equation}

Now, we set $\eta' = \frac{4\alpha^2}{(\alpha-1)^{1-1/\alpha}}$ and $\lambda =
4^{\alpha-1}(\alpha-1)^{1-1/\alpha}$. Substituting \ineqref{phi} as well as the bounds for
$\frac{dG_{\NE}(\jobset(t))}{dt}$, $\frac{dG_{\OPT}(\jobset^*(t))}{dt}$ and
$\frac{dG_1^*(\jobset(t))}{dt}$ into the running condition, we can see that it can be satisfied for
all valid values of $x_t$ by setting the multipliers to be $c_1 =
\max\{\frac{4\alpha^3}{(\alpha-1)^{2}}, 4^{\alpha}\alpha H_P\}$ and $c_2 = 2\alpha \cdot
\paren{2H_P}^{1/\alpha}$. Since $\alpha$ can be considered as a constant with respect to $P$, and
it is well-known that $H_P = O(\ln P)$, the theorem is proved.
\end{proof}

\subsubsection{Performance for Jobs with Identical Release Time}

We now bound the performance of \NEQUI for jobs with identical release time. We show that the
competitive ratio of \NEQUI can be slightly improved in this case compared to the one achieved for
arbitrarily released jobs. The following theorem gives the result.

\begin{theorem}\label{thm:batch}
\NEQUI is $O(\ln^{1/\alpha}P)$-competitive with respect to total flow time plus energy for any set
of parallel jobs with identical release time, where $P$ is the total number of processors.
\end{theorem}

\begin{proof}
We will show that the execution of any job set scheduled by \NEQUI satisfies the running condition
$\frac{dG_{\NE}(\jobset(t))}{dt} \le c_1 \cdot \frac{dG_1^*(\jobset(t))}{dt} + c_2 \cdot
\frac{dG_2^*(\jobset(t))}{dt}$,
where $c_1 = O(\ln^{1/\alpha}P)$ and $c_2 = O(\ln^{1/\alpha}P)$. The theorem is then implied by
\lemref{local}.

We first note from the proof of \thmref{release} that the instantaneous cost of \NEQUI satisfies
$\frac{dG_{\NE}(\jobset(t))}{dt} \le \frac{\alpha}{\alpha-1}n_t$, and the rate of change for the
first lower bound satisfies $\frac{dG_1^*(\jobset(t))}{dt} \ge
\frac{\alpha}{\alpha-1}\paren{\frac{1}{2H_P}}^{1/\alpha}(1-x_t)n_t$. It remains to bound the rate
of change $\frac{dG_2^*(\jobset(t))}{dt}$ for the second lower bound. To this end, we focus on the
$t$-prefix $\jobset(\la{t})$ of the job set $\jobset$ and redefine $n_t(z)$ to be the number of
jobs whose completed work is at least $z$ at time $t$ under \NEQUI, i.e., $n_t(z) = \sum_{J_i \in
\jobset}[w(J_i(\la{t})) \ge z]$. Note that we do not restrict a job to be active when considering
its contribution towards $n_t(z)$. This is consistent with the definition of $t$-prefix
$\jobset(\la{t})$ for job set $\jobset$, which allows us to express the the second lower bound
given in \lemref{lowerbounds} for the $t$-prefix $\jobset(\la{t})$ of job set $\jobset$ as follows
\begin{equation*}\label{eq:g2}
G_2^*(\jobset(\la{t})) = \frac{\alpha}{\paren{(\alpha-1)P}^{1-1/\alpha}}
\int_{0}^{\infty}\paren{\sum_{i=1}^{n_t(z)}i^{1-1/\alpha}}dz.
\end{equation*}

Again we only focus on the set $\jobset_D(t)$ of deprived jobs, and for each deprived job $J_i\in
\jobset_D(t)$, its execution rate satisfies $\Gamma_i^k(t) \ge
\paren{\frac{1}{(\alpha-1)H_P}}^{1/\alpha}\frac{a_i(t)^{1-1/\alpha}}{2^{1/\alpha}}$
by \lemref{ratebounds}. In the worst case, the $n_t^D$ deprived jobs have completed the most work
so far, so they will take the smallest $n_t^D$ indices in the above expression of
$G_2^*(\jobset(\la{t}))$. Thus, the rate of change $\frac{dG_2^*(\jobset(t))}{dt}$ is bounded by
\begin{eqnarray}\label{ineq:g2change}
\frac{dG_2^*(\jobset(t))}{dt} &=& \frac{\alpha}{\paren{(\alpha-1)P}^{1-1/\alpha}}\cdot \frac{d}{dt} \int_{0}^{\infty}\paren{\sum_{i=1}^{n_{t}(z)}i^{1-1/\alpha}} dz \nonumber \\
&=& \frac{\alpha}{\paren{(\alpha-1)P}^{1-1/\alpha}}\cdot \frac{1}{\dt} \int_{0}^{\infty}\left[\paren{\sum_{i=1}^{n_{t+\dt}(z)}i^{1-1/\alpha}} - \paren{\sum_{i=1}^{n_{t}(z)}i^{1-1/\alpha}} \right] dz \nonumber \\
&\ge& \frac{\alpha}{\paren{(\alpha-1)P}^{1-1/\alpha}}\sum_{i=1}^{n_t^D}i^{1-1/\alpha}\cdot \Gamma_i^k(t) \nonumber \\
&\ge& \frac{\alpha}{\alpha-1}\cdot \frac{x_t^2n_t}{4H_P^{1/\alpha}}.
\end{eqnarray}

Substituting \ineqref{g2change} as well as the bounds for $\frac{dG_{\NE}(\jobset(t))}{dt}$ and
$\frac{dG_1^*(\jobset(t))}{dt}$ into the running condition, we can see that it is satisfied for all
valid values of $x_t$ if we set $c_1 = 2^{1+1/\alpha}H_P^{1/\alpha}$ and $c_2 = 4H_P^{1/\alpha}$.
Thus, the theorem is proved.
\end{proof}

\subsection{Lower Bound for Any Non-clairvoyant Algorithm}

In this section, we prove a lower bound of $\Omega(\ln^{1/\alpha}P)$ on the competitive ratio of
any non-clairvoyant algorithm even with non-uniform speed assignments. Since this lower bound
matches the upper bound of \NEQUI for parallel jobs with identical release time, it shows that
\NEQUI is asymptotically optimal in that setting.

Before proving the lower bound, we first present a useful lemma, which gives the solution of a
minimization problem. The proof basically transforms this minimization problem into a convex
optimization problem, and solves it by applying the KKT conditions \cite{BoydVa04}.

\begin{lemma}\label{lem:kkt}
For any $P \ge 1$, $\alpha > 1$ and $b > 0$, if $\sum_{j}^{P}s_j^\alpha = b$ and $s_1 \ge s_2 \ge
\cdots \ge s_P \ge 0$, then $\max_{1\le h\le P}\frac{h^{1-1/\alpha}}{\sum_{j=1}^{h}s_j}$ is
minimized when $(s_1,s_2,\cdots,s_P)$ satisfy $\frac{h^{1-1/\alpha}}{\sum_{j=1}^{h}s_j} =
\frac{(h-1)^{1-1/\alpha}}{\sum_{j=1}^{h-1}s_j}$ for all $h = 2,\cdots,P$.
\end{lemma}

\begin{proof}
To prove this lemma, we transform the stated problem into a convex optimization problem. We then
show that our proposed solution satisfies the KKT conditions \cite{BoydVa04}, which are known to be
sufficient for the optimality of convex minimization problems. This then leads to the proof of the
lemma. First, by introducing a variable $y$, the original optimization problem can be transformed
into the following minimization problem:
\begin{eqnarray}
\mbox{minimize}&&y \nonumber \\
\mbox{subject to}&&\sum_{j=1}^{P} s_j^\alpha = b  \label{eq:nonconvex} \\
&& s_j \geq s_{j+1} \mbox{ for $ j = 1,\cdots,P-1$} \nonumber \\
&& y \geq \frac{h^{1-1/\alpha}}{\sum_{j=1}^{h} s_j} \mbox{ for $ h = 1,\cdots,P$} \nonumber
\end{eqnarray}

However, the above minimization problem is not convex because its equality constraint (Equation
\ref{eq:nonconvex}) is not linear.  Substituting $z_j = s_j^\alpha$, we transform it into a convex
optimization problem as follows.
\begin{eqnarray}
\mbox{minimize}&&y \nonumber \\
\mbox{subject to}&& \sum_{j=1}^{P} z_j = b   \label{eq:linear}\\
&& z_{j+1} - z_j \leq 0 \mbox{ for $ j = 1,\cdots,P-1$}  \label{ineq:greater} \\
&& \frac{h^{1-1/\alpha}}{\sum_{j=1}^{h} z_j^{1/\alpha}} - y \leq 0 \mbox{ for $ h = 1,\cdots,P$}
\label{ineq:yconst}
\end{eqnarray}

For this minimization problem, the objective function and the only equality constraint
(\eqref{linear}) are linear, the inequality constraints (\ineqreff{greater} and \ineqreff{yconst})
are convex. Note that \ineqreff{yconst} are convex because $1/f(x)$ is a convex function if $f(x)$
is a positive concave function, and ${\sum_{j=1}^{h} z_j^{1/\alpha}}$ is concave because
$z_j^{1/\alpha}$ is concave for $\alpha > 1$. We have now transformed our min-max optimization
problem into a convex minimization problem. We will prove that the following $(y^*,
z_1^*,\cdots,z_P^*)$ is an optimal solution to the above convex minimization problem, by showing
that it satisfies the KKT conditions.

\begin{equation}
y^* = \frac{h^{1-1/\alpha}}{\sum_{j=1}^{h} (z_j^*)^{1/\alpha}} \mbox{
  for } h = 1,\cdots,P
\label{eq:psolution}
\end{equation}

Let $x_j = j^{1-1/\alpha} - (j-1)^{1-1/\alpha}$ for $j = 1,\cdots,P$, so that $x_j > x_{j+1}$. From
\eqref{psolution} and \eqref{linear}, we get $z^*_j = b\cdot
\frac{x_j}{\sum_{i=1}^{P}x_j^{\alpha}}$ and therefore $z^*_j > z^*_{j+1}$ for $j = 1,\cdots,P-1$.

To prove $(y^*, z_1^*,\cdots,z_P^*)$ satisfies the KKT conditions, we need to show that it
satisfies primal feasibility, dual feasibility, complementary slackness, and stationarity.  It is
not hard to see that the proposed solution satisfies the primal feasibility in \eqref{linear},
\ineqreff{greater} and \ineqreff{yconst}.  Let us now associate multipliers with the constraints:
\begin{eqnarray*}
\lambda &:& \sum_{j=1}^{P} z_j = b\\
w_j &:& z_{j+1} - z_j \leq 0 \mbox{ for $ j = 1, ..., P-1$}\\
\mu_h &: & \frac{h^{1-1/\alpha}}{\sum_{j=1}^{h} z_j^{1/\alpha}} - y \leq 0 \mbox{ for $ h = 1, ...,
P$}
\end{eqnarray*}

Since we have $z^*_j > z^*_{j+1}$ for $j = 1,\cdots,P-1$, then to satisfy complementary slackness,
we have $w_j = 0$ for $j=1,\cdots,P-1$.  Now we need to show that there exists $\lambda$ and $\mu_h
\geq 0$ such that dual feasibility and stationarity are satisfied. To derive the stationarity
condition, consider the Lagrangian function:
\begin{equation}
L(y, z_j, \lambda, \mu_h) = y + \sum_{h = 1}^{P} \mu_h
\left(\frac{h^{1-1/\alpha}}{\sum_{j=1}^{h}z_j^{1/\alpha}} - y \right) +
\lambda\left(\sum_{j=1}^{P}z_j - b\right). \nonumber
\end{equation}

Taking the derivative of the Lagrangian function with respect to $y$ and $z_j$, setting them to
zero, and substituting in $(y^*, z_1^*, ... , z_P^*)$, we get the following set of stationarity
conditions:
\begin{eqnarray}
\sum_{h = 1}^{P}{\mu_h} &=& 1, \label{eq:sum} \\
\frac{(y^*)^2}{\alpha (z_j^*)^{1-1/\alpha}} \left(\sum_{h = j}^P
\frac{\mu_h}{h^{1-1/\alpha}}\right) &=& \lambda \mbox{ for } j= 1,\cdots,P. \label{eq:mu}
\end{eqnarray}

Solving the linear system in \eqreff{mu} by considering $\mu_h$ as variables, we have $\mu_h = c_h
\cdot \lambda$, where $c_h = \frac{h^{1-1/\alpha}\left(\left(z_h^*\right)^{1-1/\alpha} -
  \left(z_{h+1}^*\right)^{1-1/\alpha}\right)\alpha }{(y^*)^2} $, for
each $h=1,\cdots,P$, and $z_{P+1}^*$ is defined to be $0$. According to the values of $(y^*, z_1^*,
\cdots, z_P^*)$, we know that $y^*>0$, $z_h^* > 0$ and $z_h^* > z_{h+1}^*$. Therefore, we have $c_h
> 0$ for $h = 1, ..., P$.  Substituting $\mu_h = c_h \cdot \lambda$ into \eqref{sum}, we get
$\lambda = \frac{1}{\sum_{h=1}^{P} c_h} > 0$, which implies that $\mu_h > 0$ for all $h =
1,\cdots,P$. Thus, we have shown that the dual feasibility is satisfied. Moreover, there exists
$\lambda$ and $\mu_h$ that make our proposed solution $(y^*, z_1^*, \cdots, z_P^*)$ satisfy
stationarity, and hence all the KKT conditions. Therefore, it is an optimal solution for the convex
minimization problem, and the corresponding speed assignment $s_j^* = (z_j^*)^{1/\alpha}$ is
optimal for the original optimization problem.
\end{proof}

Using \lemref{kkt}, the following theorem gives the lower bound for any non-clairvoyant algorithm.

\begin{theorem}
Any non-clairvoyant algorithm is $\Omega(\ln^{1/\alpha}P)$-competitive with respect to total flow
time plus energy, where $P$ is the total number of processors.
\end{theorem}

\begin{proof}
Consider a job set $\jobset$ containing only a single job with constant parallelism $h$ and work
$w$, where $1 \le h \le P$ and $w > 0$. For any non-clairvoyant algorithm \A, we can assume without
loss of generality that it allocates all $P$ processors to the job with speeds $s_1 \ge s_2 \ge
\cdots \ge s_P \ge 0$, which do not change throughout the job's execution since the work $w$ can be
arbitrarily small. The power consumption of \A at any time is then given by $u =
\sum_{j=1}^{P}s_j^\alpha$. The flow time plus energy of $\jobset$ scheduled by \A is $G_A(\jobset)
= \paren{1 + u}\frac{w}{\sum_{j=1}^{h}s_j}$. The optimal offline scheduler, knowing the parallelism
$h$, will allocate exactly $h$ processors of speed $\paren{\frac{1}{(\alpha-1)h}}^{1/\alpha}$, thus
incurring flow time plus energy of $G_{\OPT}(\jobset) = \frac{\alpha}{(\alpha-1)^{1-1/\alpha}}\cdot
\frac{w}{h^{1-1/\alpha}}$. The competitive ratio of \A is
$\frac{G_{\A}(\jobset)}{G_{\OPT}(\jobset)} = \frac{(\alpha-1)^{1-1/\alpha}(1+u)}{\alpha}\cdot
\frac{h^{1-1/\alpha}}{\sum_{j=1}^{h}s_j}$.

The adversary will choose parallelism $h$ to maximize this ratio, i.e., to find $\max_{1\le h\le
  P}\frac{G_{\A}(\jobset)}{G_{\OPT}(\jobset)}$, while the online algorithm
\A chooses $(s_1,\cdots,s_P)$ to minimize $\max_{1\le h\le
  P}\frac{G_{\A}(\jobset)}{G_{\OPT}(\jobset)}$ regardless of the choice of
$h$. According to \lemref{kkt}, $\max_{1\le h\le
  P}\frac{h^{1-1/\alpha}}{\sum_{j=1}^{h}s_j}$ is minimized when
$\frac{h^{1-1/\alpha}}{\sum_{j=1}^{h}s_j} = \frac{(h-1)^{1-1/\alpha}}{\sum_{j=1}^{h-1}s_j}$ for $h
= 2,\cdots,P$. Hence, by solving this set of equations, the best non-clairvoyant algorithm will set
$s_j = \paren{j^{1-1/\alpha} -
  (j-1)^{1-1/\alpha}}s_1$ for $j = 1, 2, \cdots, P$. Since
$j^{1-1/\alpha} - (j-1)^{1-1/\alpha} \ge \frac{1-1/\alpha}{j^{1/\alpha}}$, we have $s_j \ge
\frac{1-1/\alpha}{j^{1/\alpha}}s_1$.  Substituting these into $u = \sum_{j=1}^{P}s_j^\alpha$, we
get $s_1 \le \frac{\alpha
  u^{1/\alpha}}{(\alpha-1)H_P^{1/\alpha}}$, where $H_P$ is the $P$'th
Harmonic number. The competitive ratio of any non-clairvoyant algorithm \A thus satisfies
$\frac{G_{\A}(\jobset)}{G_{\OPT}(\jobset)} \ge \frac{(\alpha-1)^{1-1/\alpha}(1+u)}{\alpha}\cdot
\frac{1}{s_1} \ge \frac{(\alpha-1)^{2-1/\alpha}}{\alpha^2} \cdot
\frac{1+u}{u^{1/\alpha}}H_P^{1/\alpha} \ge \frac{\alpha-1}{\alpha}\cdot H_P^{1/\alpha}$. The last
inequality holds because $\frac{1+u}{u^{1/\alpha}}$ is minimized when $u = \frac{1}{\alpha-1}$.
Since it is also known that $H_P = \Omega(\ln P)$, the theorem is proved.
\end{proof}

\subsection{IP-clairvoyant Algorithm: \UCEQ}

In this section, we present an IP-clairvoyant algorithm called \UCEQ (Uniform Conservative
Equi-Partitioning). We show that knowledge about the instantaneous parallelism of the jobs does
help to improve the performance of an online algorithm. In particular, we prove that \UCEQ is
$O(1)$-competitive with respect to total flow time plus energy, even for jobs with arbitrary
release time.

\algref{uceq} describes the \UCEQ algorithm. As we can see, \UCEQ works similarly to \NEQUI, but it
never allocates more processors to a job than the job's instantaneous parallelism at any time.
Moreover, the speeds of the processors allocated to a job are assigned in a uniform manner.

\begin{algorithm}
\caption{\UCEQ}\label{alg:uceq}
\begin{algorithmic}[1]
\REQUIRE{total number $P$ of processors, number $n_t$ of active jobs at time $t$ and the
instantaneous parallelism $h_i^k$ of each active job $J_i$ at time $t$. } \ENSURE{number of
allocated processors and their speeds for each active job at time $t$.}

\IF {$P \ge n_t$} \STATE{allocate $a_i(t) = \min\{h_i^k, \floor{\frac{P}{n_t}}\}$ processors to
each active job $J_i$.} \STATE {set the speed of all $a_i(t)$ processors allocated to job $J_i$ to
be $s_i(t) =
\paren{\frac{1}{(\alpha-1)a_i(t)}}^{1/\alpha}$.} \ELSE
\STATE{allocate $a_i(t) = \frac{P}{n_t}$ fraction of a processor to each active job $J_i$.}
\STATE{set the speed of all processors to be $s(t) =
\paren{\frac{n_t}{(\alpha-1)P}}^{1/\alpha}$.} \ENDIF
\end{algorithmic}
\end{algorithm}

To analyze the performance of \UCEQ, we say that job $J_i$ is \emph{satisfied} at time $t$ if
$a_i(t) = h_i^k$, and that it is \emph{deprived} if $a_i(t) < h_i^k$. We see that at time $t$, a
job $J_i$ scheduled by \UCEQ has execution rate $\Gamma_i^k(t) =
\frac{a_i(t)^{1-1/\alpha}}{(\alpha-1)^{1/\alpha}}$ and consumes power $u_i(t) =
\frac{1}{\alpha-1}$. Therefore, the overall power consumption is given by $u_t =
\frac{n_t}{\alpha-1}$. The following theorem gives the performance of \UCEQ for jobs with arbitrary
released times.

\begin{theorem}\label{thm:iprelease}
\UCEQ is $\paren{\max\{\frac{4\alpha^2}{\alpha-1}, 4^{\alpha}\alpha\} + 2\alpha}$-competitive with
respect to total flow time plus energy for any set of parallel jobs.
\end{theorem}

\begin{proof}
As with the analysis of \NEQUI, we prove the competitiveness of \UCEQ using amortized local
competitiveness argument with the same potential function as in \eqref{potential}, but with $\eta$
now set to $\eta = \frac{\eta'}{P^{1-1/\alpha}}$, where $\eta' =
\frac{4\alpha^2}{(\alpha-1)^{1-1/\alpha}}$.  Clearly, the boundary, arrival and completion
conditions continue to hold. We now show that the execution of any job set under \UCEQ (\UC for
short) satisfies the running condition $\frac{dG_{\UC}(\jobset(t))}{dt} + \frac{d\Phi(t)}{dt} \le
c_1 \cdot \frac{dG_{\OPT}(\jobset^*(t))}{dt} + c_2 \cdot \frac{dG_1^*(\jobset(t))}{dt}$, where $c_1
= \max\{\frac{4\alpha^2}{\alpha-1}, 4^{\alpha}\alpha\}$ and $c_2 = 2\alpha$.

Following the proof of \thmref{release}, we have $\frac{dG_{\UC}(\jobset(t))}{dt} =
\frac{\alpha}{\alpha-1}n_t$, $\frac{dG_{\OPT}(\jobset^*(t))}{dt} = n_t^* +
\sum_{j=1}^{P}\paren{s_j^*}^{\alpha}$, and $\frac{dG_1^*(\jobset(t))}{dt} \ge
\frac{\alpha}{(\alpha-1)^{1-1/\alpha}} \sum_{J_i\in
  \jobset_S(t)}\frac{\Gamma_i^{k}(t)}{\paren{h_i^k}^{1-1/\alpha}} =
\frac{\alpha}{\alpha-1}(1-x_t)n_t$, where $x_t$ is the deprived ratio as defined earlier. The rate
of change $\frac{d\Phi(t)}{dt}$ for the potential function $\Phi(t)$ at time $t$ can be shown to
satisfy
\begin{eqnarray}\label{ineq:phi2}
\frac{d\Phi(t)}{dt} &\le& \frac{\eta'}{P^{1-1/\alpha}}
\paren{-\sum_{i=1}^{n_t^D}i^{1-1/\alpha}\cdot \Gamma_i^{k}(t) + n_t^{1-1/\alpha}
\sum_{j=1}^{P}s_j^* + n_t^* \sum_{i=1}^{n_t}\paren{i^{1-1/\alpha} - (i-1)^{1-1/\alpha}}
\Gamma_i^{k}(t)} \nonumber \\
&\le& \eta' \paren{-\frac{x_t^2}{4(\alpha-1)^{1/\alpha}}n_t +
\frac{\lambda}{\alpha}\sum_{j=1}^{P}\paren{s_j^*}^{\alpha} +
\frac{1-1/\alpha}{\lambda^{1/(\alpha-1)}}n_t + \frac{n_t^*}{(\alpha-1)^{1/\alpha}}}, \nonumber
\end{eqnarray}
where $\lambda = 4^{\alpha-1}(\alpha-1)^{1-1/\alpha}$. Substituting these bounds into the running
condition, we see that it is satisfied for all valid values of $x_t$. Hence, the theorem is proved.
\end{proof}

\thmref{iprelease} shows that \UCEQ is $O(1)$-competitive for total flow time plus energy, since
$\alpha$ can be considered as a constant with respect to $P$. The competitive ratio, however, is
still exponential in $\alpha$. The following theorem shows that, for jobs with identical release
time, the competitive ratio can be further improved to be strictly smaller than $6$ regardless of
the value of $\alpha$.

\begin{theorem}\label{thm:ipbatch}
\UCEQ is $\paren{2^{2-1/\alpha} + 2}$-competitive with respect to total flow time plus energy for
any set of parallel jobs with identical release time.
\end{theorem}

\begin{proof}
Similarly to the proof for \NEQUI, we prove the competitiveness of \UCEQ for jobs with identical
release time using the local competitive argument. In particular, we show that any job set
scheduled by \UCEQ satisfies the running condition $\frac{dG_{\UC}(\jobset(t))}{dt} \le 2 \cdot
\frac{dG_1^*(\jobset^*(t))}{dt} + 2^{2-1/\alpha} \cdot \frac{dG_2^*(\jobset(t))}{dt}$. Again, we
have $\frac{dG_{\UC}(\jobset(t))}{dt} = \frac{\alpha}{\alpha-1}n_t$ and
$\frac{dG_1^*(\jobset(t))}{dt} \ge \frac{\alpha}{\alpha-1}(1-x_t)n_t$ from the proof of
\thmref{iprelease}. Also, by following the proof of \thmref{batch}, we can get
$\frac{dG_2^*(\jobset(t))}{dt} \ge
\frac{\alpha}{\paren{(\alpha-1)P}^{1-1/\alpha}}\sum_{i=1}^{n_t^D}i^{1-1/\alpha}\cdot \Gamma_i^k(t)
\ge \frac{\alpha}{\alpha-1}\cdot \frac{x_t^2n_t}{2^{2-1/\alpha}}$.  Using these bounds, the desired
running condition can be satisfied for all valid values of $x_t$.
\end{proof}

\UCEQ significantly improves upon any non-clairvoyant algorithm with respect to total flow time
plus energy.  This is essentially due to \UCEQ not wasting any energy yet still guaranteeing a
sufficient execution rate for the jobs. Since it is known that non-clairvoyant algorithms perform
similarly to the IP-clairvoyant ones with respect to total flow time alone
\cite{DengGuBr00,EdmondsChBr03,KalyanasundaramPr00,EdmondsPr09}, \thmreftwo{iprelease}{ipbatch}
show the importance of even partial clairvoyance when energy is also of concern. Moreover, since
\UCEQ is aware of the instantaneous parallelism of the jobs, uniform speed scaling is sufficient to
ensure its competitiveness. Therefore, compared to the non-clairvoyant algorithm \NEQUI, which
requires non-uniform speed scaling, \UCEQ may be more feasible in practice, especially when
scheduling for jobs whose parallelism does not change frequently. Lastly, note that non-uniform
speed scaling is not beneficial in the IP-clairvoyant setting. Indeed, using non-uniform speeds can
only degrade an algorithm's performance, since generally less energy will be consumed at the same
execution rate when using a uniform speed \cite{YaoDeSh95}.

\section{Makespan Plus Energy}\label{sec:makespan}

In this section, we consider the objective of minimizing makespan plus energy. In particular, we
propose an IP-clairvoyant algorithm called \PF, and show that it is
$O(\ln^{1-1/\alpha}P)$-competitive for any set of parallel jobs regardless of their release time.
We also show that this ratio is asymptotically optimal for any IP-clairvoyant algorithm.

\subsection{Performance of the Optimal}

To bound the performance of \PF, we first derive a bound on the performance of the optimal offline
scheduler. It turns out that, for makespan plus energy, we need only focus on the case where all
jobs are released together, as we will show in \secref{makespanrelease}, adding release time to the
jobs will increase the competitive ratio of our algorithm by at most a constant factor in the worst
case.

In the following lemma, we show that to minimize makespan plus energy for jobs with identical
release time, the optimal scheduler always maintains a constant total power of $\frac{1}{\alpha-1}$
at any time. A similar property was proven in \cite{PruhsVaUt08} for the makespan minimization
problem with a total energy budget.

\begin{lemma}\label{lem:equalpower}
Given a set of jobs with identical release time, the optimal scheduler will execute the jobs with a
constant total power of $\frac{1}{\alpha-1}$ at any time during the execution.
\end{lemma}

\begin{proof}
We prove the lemma by contradiction.

Consider an interval $\dt$ during which the speeds of all processors, denoted by $(s_1, s_2,
\cdots, s_P)$, remain unchanged in the optimal schedule. The makespan plus energy incurred when
executing this portion of the job set is given by $H = \dt (1 + u)$, where $u =
\sum_{j=1}^{P}s_j^{\alpha}$ is the power consumption of all the processors during $\dt$. Suppose
that $u \neq \frac{1}{\alpha-1}$. We will show that by modifying the power consumption, we can
reduce the overall makespan plus energy.

Specifically, the modified schedule executes the same portion of the job set by running the $j$'th
processor at speed $k\cdot s_j$, where $k = \paren{\frac{1}{(\alpha-1)u}}^{1/\alpha}$. This portion
will then finish in $\frac{\dt}{k}$ time, and consumes $\frac{1}{\alpha-1}$ power at any time
during this interval. The new makespan plus energy incurred when executing this portion of the job
set is $H' = \frac{\dt}{k} (1 + \frac{1}{\alpha-1}) = \frac{\alpha}{(\alpha-1)^{1-1/\alpha}}\dt
u^{1/\alpha}$. Hence, we have $\frac{H}{H'} = \frac{(\alpha-1)^{1-1/\alpha}}{\alpha} \cdot
\frac{1+u}{u^{1/\alpha}} > 1$, i.e., $H>H'$, since $\frac{1+x}{x^{1/\alpha}}$ is uniquely minimized
at $x = \frac{1}{\alpha-1}$ for all $x>0$ and $u \neq \frac{1}{\alpha-1}$. While the costs incurred
when executing other portions of the job set are unchanged, the modified schedule incurs a strictly
less makespan plus energy. This contradicts the fact that the original schedule is optimal.
\end{proof}

Intuitively, since the jobs contribute $1$ towards the makespan part of the objective function at
any time during their execution, \lemref{equalpower} implies that the optimal strategy provides a
balanced contribution towards the power consumption part of the objective at all time.\footnote{A
similar argument will show that the optimal strategy for the objective of total flow time plus
energy maintains a power consumption of $\frac{n_t}{\alpha-1}$ at any time $t$, where $n_t$ is the
number of active jobs at time $t$. \lemref{lowerbounds} shows that the optimal scheduler indeed
satisfies this property.} However, this result only holds when the jobs have identical release
time. For jobs with arbitrary release time, speeding up the execution whenever the power is less
than $\frac{1}{\alpha-1}$ may not be helpful, for the subsequent jobs may not have arrived yet to
fill in the gap. Therefore, the optimal power consumption in this case would be upper-bounded by
$\frac{1}{\alpha-1}$.

\lemref{equalpower} shows that the total power consumed by all the jobs should be constant over
time in the optimal schedule, provided that all jobs are released at the same time. The following
lemma shows that the optimal scheduler also uses a constant power throughout the execution of any
individual job, provided that sufficient processors are available.

\begin{lemma}\label{lem:equalphase}
Given a set of jobs, suppose that there are sufficient processors available to satisfy all jobs at
all times, i.e., $P \ge \sum_{i = 1..n} h_i^{max}$, where $h_i^{max}=\max_{k = 1..k_i}h_i^k$. Then,
the optimal scheduler will allocate a constant power to any individual job throughout its execution
lifetime.
\end{lemma}

\begin{proof}
Again, we prove the lemma by contradiction.

Since there are sufficient processors in the system, by the convexity of the power function, the
optimal scheduler should allocate exactly $h_i^k$ processors of the same speed to each phase
$J_i^k$. Now, suppose that there exist two phases from the same job, to which the optimal scheduler
does not allocate the same power, and let $\langle w_1, h_1\rangle$ and $\langle w_2, h_2\rangle$
denote the work-parallelism pairs of these two phases, respectively. Thus, we have $h_1s_1^\alpha
\neq h_2s_2^\alpha$, where $s_1$ and $s_2$ denote the speeds of the processors allocated to the two
phases. We will show that, by modifying the power allocations for the two phases, we can reduce the
overall energy consumption while maintaining the execution time of the job.

Let $t_1 = \frac{w_1}{h_1s_1}$ and $t_2 = \frac{w_2}{h_2s_2}$. The overall execution time and
energy consumption of the two phases are given by $T = t_1 + t_2$ and $E = t_1\cdot h_1s_1^\alpha +
t_2\cdot h_2s_2^\alpha$, respectively. Let $s'_1$ and $s'_2$ denote the speeds used for the two
phases in the modified schedule, and we will make sure that their power consumptions are identical,
i.e.,
\begin{equation}\label{eq:one}
h_1(s'_1)^\alpha = h_2(s'_2)^\alpha = u.
\end{equation}
Moreover, to maintain the same execution time for the job, the processor speeds in the modified
schedule should also satisfy $\frac{w_1}{h_1s'_1} + \frac{w_2}{h_2s'_2} = t_1 + t_2$. This gives us
\begin{equation}\label{eq:two}
1 = \beta\cdot \frac{s_1}{s'_1} + (1-\beta)\cdot \frac{s_2}{s'_2},
\end{equation}
where $\beta = \frac{t_1}{t_1+t_2}$. Solving \eqreftwo{one}{two}, we get $u = \paren{\beta\cdot
h_1^{1/\alpha} s_1 + (1-\beta) \cdot
  h_2^{1/\alpha} s_2}^\alpha$, and $s'_1 =
\paren{\frac{u}{h_1}}^{1/\alpha}$ and $s'_2 =
\paren{\frac{u}{h_2}}^{1/\alpha}$. The total energy consumption for the two phases in the modified schedule
is then given by
\begin{eqnarray*}
E' &=& \paren{\frac{w_1}{h_1s'_1} +
  \frac{w_2}{h_2s'_2}}u \\
  &=& \paren{t_1\cdot \frac{s_1}{s'_1} + t_2\cdot
  \frac{s_2}{s'_2}}u \\
  &=& \paren{t_1\cdot h_1^{1/\alpha}s_1 + t_2\cdot
  h_2^{1/\alpha}s_2}u^{1-1/\alpha} \\
  &=& (t_1+t_2)\paren{\beta\cdot
  h_1^{1/\alpha} s_1 + (1-\beta) \cdot h_2^{1/\alpha} s_2}^\alpha \\
  &<& (t_1+t_2)\paren{\beta\cdot h_1 s_1^\alpha + (1-\beta) \cdot
  h_2s_2^\alpha} \\
  &=& t_1\cdot h_1s_1^\alpha + t_2\cdot h_2s_2^\alpha \\
  &=& E.
\end{eqnarray*}

The inequality is because $x^\alpha$ is strictly convex for $\alpha>1$, and $h_1^{1/\alpha}s_1 \neq
h_2^{1/\alpha}s_2$ according to our assumption. Hence, the modified schedule consumes strictly less
energy while having the same makespan. This contradicts the fact that the original schedule is
optimal.
\end{proof}

We now give lower bounds on the performance of the optimal offline scheduler. In contrast to the
total flow time plus energy, where the completion time of each job contributes to the overall
objective function, the makespan for a set of jobs only depends on the completion time of the last
job. Hence, the other jobs only contribute to the energy consumption part of the objective, and can
therefore be slowed down to consume less energy and improve the overall performance. Based on this
observation as well as \lemreftwo{equalpower}{equalphase}, we derive the following two lower
bounds.

\begin{lemma}\label{lem:optimalH}
The optimal makespan plus energy for any set $\jobset$ of jobs with identical release time
satisfies $H_{\OPT}(\jobset) \ge \max\{H_1^*(\jobset), H_2^*(\jobset)\}$, where
\begin{eqnarray}
H_1^*(\jobset) &=& \frac{\alpha}{(\alpha-1)^{1-1/\alpha}}
\cdot\frac{\sum_{i=1}^{n}w(J_{i})}{P^{1-1/\alpha}},
\\
H_2^*(\jobset) &=& \frac{\alpha}{(\alpha-1)^{1-1/\alpha}}\cdot
\paren{{\sum_{i=1}^{n}x(J_{i})^{\alpha}}}^{1/\alpha},
\end{eqnarray}
and where $w(J_i)$ and $x(J_i)$ denote the work and the unit-time span of job $J_i$ respectively,
and $P$ is the total number of processors.
\end{lemma}

\begin{proof}
By \lemref{equalpower}, the optimal scheduler at any time $t$ consumes power $u_t^* =
\frac{1}{\alpha-1}$. Hence, the energy consumption $E_{\OPT}(\jobset)$ under the optimal schedule
satisfies $E_{\OPT}(\jobset) = \int_{0}^{M_{\OPT}(\jobset)}u_t^*dt =
\frac{1}{\alpha-1}M_{\OPT}(\jobset)$. The optimal makespan plus energy is then given by
$H_{\OPT}(\jobset) = M_{\OPT}(\jobset) + E_{\OPT}(\jobset) =
\frac{\alpha}{\alpha-1}M_{\OPT}(\jobset)$. Therefore, we only focus on the makespan in the
following.

To get the first lower bound, we observe that the maximum execution rate on $P$ processors using a
total power of $\frac{1}{\alpha-1}$ is achieved when all processors run at the same speed of $s =
\paren{\frac{1}{(\alpha-1)P}}^{1/\alpha}$. Since the total work of all
jobs in the job set is $\sum_{i=1}^{n}w(J_{i})$, the optimal makespan satisfies $M_{\OPT}(\jobset)
\ge  \frac{\sum_{i=1}^{n}w(J_{i})}{Ps} =  (\alpha-1)^{1/\alpha}\cdot
\frac{\sum_{i=1}^{n}w(J_{i})}{P^{1-1/\alpha}}$.

For the second lower bound, we give infinite number of processors to the optimal scheduler, which
can only improve its performance. Thus, all jobs can be satisfied at all times, and the optimal
scheduler will allocate exactly $h_{i}^{k}$ processors of the same speed $s_{i}^k$ to each phase
$J_i^k$. By \lemref{equalphase}, the power consumption of any job is constant over all phases,
i.e., for each job $J_i\in\jobset$, we have $h_i^1(s_i^1)^\alpha = h_i^k(s_i^k)^\alpha$ for all $1
\le k \le k_i$. Hence, the execution rate for phase $J_i^k$ is given by $h_i^ks_i^k =
(h_i^k)^{1-1/\alpha}(h_i^1)^{1/\alpha}s_i^1$, and so the completion time of job $J_i$  satisfies
$c_i = \sum_{k=1}^{k_i} \frac{w_i^k}{h_i^ks_i^k} = \frac{x(J_i)}{(h_i^1)^{1/\alpha}s_i^1}$.
Moreover, the optimal scheduler will finish all the jobs at the same time, since otherwise slowing
down the jobs that finish early will reduce the energy without increasing makespan. This gives
$\frac{x(J_1)}{(h_1^1)^{1/\alpha}s_1^1} = \frac{x(J_i)}{(h_i^1)^{1/\alpha}s_i^1}$ for all $1\le
i\le n$.  By \lemref{equalpower}, the total power is constant when the optimal scheduler starts to
execute the first phase of all jobs, i.e., $\sum_{i=1}^{n}h_i^1(s_i^1)^{\alpha} =
\frac{1}{\alpha-1}$. Solving all these equations gives us $s_i^1 =
\paren{\frac{1}{\alpha-1}}^{1/\alpha} \cdot \frac{x(J_i)}{(h_i^1)^{1/\alpha}}\cdot
\frac{1}{\paren{\sum_{i=1}^{n}x(J_i)^\alpha}^{1/\alpha}}$.  The optimal makespan is thus
$M_{\OPT}(\jobset) = c_i = \frac{x(J_i)}{(h_i^1)^{1/\alpha}s_i^1} = (\alpha-1)^{1/\alpha}\cdot
\paren{\sum_{i=1}^{n}x(J_i)^\alpha}^{1/\alpha}$.
\end{proof}

\subsection{IP-clairvoyant Algorithm: \PF}

We now present an IP-clairvoyant algorithm \PF (Work-Conserving Equal-Power) for any set of
parallel jobs. \algref{PF} describes its details.

There are two operating modes in \PF, namely Work-Conserving (WC) and Equal-Power (EP). Whenever
the sum of the instantaneous parallelism of all active jobs exceeds the total number of processors,
the algorithm enters WC mode. In this mode, the $P$ processors can be allocated in any manner as
long as all processors are assigned and no job receives more processors than its instantaneous
parallelism. This can be achieved by any work-conserving algorithm, such as Proportional Allocation
\cite{BrechtGu96} or Dynamic Equi-partitioning \cite{BrechtDeGu95,DengGuBr00}. All processors in
this mode share the same speed. On the other hand, when the total instantaneous parallelism of all
active jobs is not more than the total number of processors, the algorithm enters EP mode, in which
each job receives exactly the same number of processors as its instantaneous parallelism and all
jobs consumes the same power.

\begin{algorithm}
\caption{\PF}\label{alg:PF}
\begin{algorithmic}[1]\normalsize
\REQUIRE{total number $P$ of processors and the instantaneous parallelism $h_i^k$ of each active
job $J_i \in \jobset(t)$ at time $t$. } \ENSURE{number of allocated processors and their speeds for
each active job at time $t$.}

\IF {$\sum_{J_i\in \jobset(t)}h_i^k > P$}\STATE {allocate $a_i(t) \le h_i^k$ processors to each job
$J_i \in \jobset(t)$ subject to $\sum_{J_i\in \jobset(t)}a_i(t) = P$.} \STATE{set the speed of all
$P$ processors to be $s(t) =
\paren{\frac{1}{(\alpha-1)P}}^{1/\alpha}$.} \ELSE \STATE {allocate $a_i(t) = h_i^k$ processors to each
job $J_i \in \jobset(t)$;} \STATE{set the speed of all $a_i(t)$ processors allocated to job $J_i$
to be $s_i(t)=\paren{\frac{1}{(\alpha-1)a_i(t)n_t}}^{1/\alpha}$.} \ENDIF
\end{algorithmic}
\end{algorithm}

Similarly to the optimal offline scheduler, we can see from \algref{PF} that \PF consumes a total
power of $u_t = \frac{1}{\alpha-1}$ at any time $t$. Hence, the overall energy consumption
satisfies $E(\jobset) = \frac{1}{\alpha-1}M(\jobset)$, and the makespan plus energy is given by
$H(\jobset) = \frac{\alpha}{\alpha-1}M(\jobset)$. For convenience, we drop the algorithm subscript
in this section and let $H(\jobset)$ denote $H_{\PF}(\jobset)$, since \PF is the only algorithm we
study for this objective.

\subsubsection{Performance for Jobs with Identical Release Time}\label{sec:makespanbatch}

Before analyzing the performance of \PF for jobs with identical release time, we first define some
notations and prove a useful lemma. For each job $J_i$, let $J_i^{wc}$ and $J_i^{ep}$ denote the
portions of the job executed under WC and EP modes, respectively. Moreover, we define $\jobset^{wc}
= \{J_i^{wc}: J_i\in \jobset\}$ and $\jobset^{ep} = \{J_i^{ep} : J_i\in \jobset\}$. The following
lemma shows that all jobs in $\jobset^{ep}$ reduce their unit-power span at the same rate.

\begin{lemma}\label{lem:spanreduce}
Suppose that \PF schedules a set of jobs with identical release time. Then at any time $t$ in EP
mode, the unit-power span of any active job in $\jobset^{ep}$ is reduced at the rate of
$\paren{\frac{1}{(\alpha-1)n_t}}^{1/\alpha}$ regardless of its instantaneous parallelism, where
$n_t$ is the total number of active jobs at time $t$.
\end{lemma}

\begin{proof}
At any time $t$ when \PF is in EP mode, let $h_i^k$ denote the instantaneous parallelism of active
job $J_i\in \jobset^{ep}$. According to \algref{PF}, the work of the job is reduced at a rate of
$\frac{dw(J_i^{ep})}{dt} = a_i(t)s_i(t) =
\frac{(h_i^k)^{1-1/\alpha}}{\paren{(\alpha-1)n_t}^{1/\alpha}}$. Hence, by definition, the
unit-power span of the job is reduced at a rate of $\frac{dx(J_i^{ep})}{dt} =
\frac{dw(J_i^{ep})}{dt} \cdot \frac{1}{(h_i^k)^{1-1/\alpha}} =
\paren{\frac{1}{(\alpha-1)n_t}}^{1/\alpha}$.
\end{proof}

The following theorem shows the competitive ratio of \PF for jobs with identical release time.

\begin{theorem}\label{thm:pf}
\PF is $\Theta(\ln^{1-1/\alpha}P)$-competitive with respect to makespan plus energy for any set of
jobs with identical release time, where $P$ is the total number of processors.
\end{theorem}

\begin{proof}
We again only focus on the makespan $M(\jobset)$ of the job set scheduled by \PF, since the
makespan plus energy satisfies $H(\jobset) = \frac{\alpha}{\alpha-1}M(\jobset)$. We separately
bound the time $M_{wc}(\jobset)$ when the algorithm is in WC mode and the time $M_{ep}(\jobset)$
when the algorithm is in EP mode. Obviously, we have $M(\jobset) =
M_{wc}(\jobset)+M_{ep}(\jobset)$.

We first bound $M_{wc}(\jobset)$. According to \PF, the total execution rate for the active jobs at
any time $t$ in WC mode is given by $\frac{P^{1-1/\alpha}}{(\alpha-1)^{1/{\alpha}}}$. Since the
total work of all jobs in $\jobset^{wc}$ satisfies $\sum_{i=1}^{n}w(J_{i}^{wc}) \le
\sum_{i=1}^{n}w(J_{i})$, we have $M_{wc}(\jobset) \le
(\alpha-1)^{1/\alpha}\frac{\sum_{i=1}^{n}w(J_{i})}{P^{1-1/\alpha}}$.

We now bound $M_{ep}(\jobset)$ when the algorithm is in EP mode. Let $T$ denote the first time
instance when the algorithm enters EP mode, and let $m$ denote the number of active jobs at $T$,
i.e., $m = n_T$. Since the instantaneous parallelism of each active job is at least $1$ and all $m$
active jobs are satisfied at $T$, we have $m < P$. For convenience, rename the jobs in
$\jobset^{ep}$ in non-decreasing order of their unit-power span, i.e., $x(J_1^{ep}) \le x(J_2^{ep})
\le \cdots \le x(J_m^{ep})$. According to \lemref{spanreduce}, whenever the algorithm is in EP
mode, \PF will reduce the unit-power span of all the active jobs in $\jobset^{ep}$ at the same rate
of $\paren{\frac{1}{(\alpha-1)n_t}}^{1/\alpha}$. Thus, the jobs in $\jobset^{ep}$ will complete in
exactly the above order. Let $x(J_0^{ep}) = 0$.  Then we have $M^{ep}(\jobset) =
\sum_{i=1}^{m}\frac{x(J_i^{ep}) -
  x(J_{i-1}^{ep})}{\paren{\frac{1}{(\alpha-1)(m-i+1)}}^{1/\alpha}} =
(\alpha-1)^{1/\alpha} \sum_{i=1}^{m}\paren{{(m-i+1)^{1/\alpha} -
    (m-i)^{1/\alpha}}}x(J_i^{ep})$.  For convenience, define $c_i =
(m-i+1)^{1/\alpha} - (m-i)^{1/\alpha}$ for $1 \le i \le m$, so we have $c_i \le
\frac{1}{(m-i+1)^{1-1/\alpha}}$. Let $R = \sum_{i=1}^{m}x(J_i^{ep})^{\alpha}$, and subject to this
condition and the ordering of $x(J_i^{ep})$, we can show using Lagrange multipliers that
$\sum_{i=1}^{m}c_i\cdot x(J_i^{ep})$ is maximized when $x(J_i^{ep}) = R^{1/\alpha}\cdot
c_i^{\frac{1}{\alpha-1}}/\paren{\sum_{i=1}^{m}c_i^{\frac{\alpha}{\alpha-1}}}^{1/\alpha}$. Hence, we
have $M^{ep}(\jobset) \le (\alpha-1)^{1/\alpha} R^{1/\alpha}
\paren{\sum_{i=1}^{m}c_i^{\frac{\alpha}{\alpha-1}}}^{1-1/\alpha} \le
(\alpha-1)^{1/\alpha} R^{1/\alpha} H_m^{1-1/\alpha}$, where $H_m = 1 + 1/2 + \cdots + 1/m$ is the
$m$'th harmonic number.

The makespan plus energy of the job set scheduled under \PF thus satisfies $H(\jobset) \le
\frac{\alpha}{(\alpha-1)^{1-1/\alpha}}\cdot
\paren{\frac{\sum_{i=1}^{n}w(J_{i})}{P^{1-1/\alpha}} + R^{1/\alpha}
  H_m^{1-1/\alpha}}$. Thus, since $\sum_{i=1}^{n}x(J_i)^{\alpha} \ge
\sum_{i=1}^{m}x(J_i^{ep})^{\alpha} = R$, we get using \lemref{optimalH} that $H(\jobset) \le (1 +
H_m^{1-1/\alpha})\cdot H_{\OPT}(\jobset) = O(\ln^{1-1/\alpha}P) \cdot H_{\OPT}(\jobset)$.

To show that this ratio is asymptotically optimal for \PF, consider a set $\jobset$ of $P$
sequential jobs released at time $0$, where the $i$'th job has unit-power span $x(J_i) =
\frac{1}{(P-i+1)^{1/\alpha}}$. From \lemref{optimalH}, the optimal scheduler has makespan plus
energy $H_{\OPT}(\jobset)= \frac{\alpha}{(\alpha-1)^{1-1/\alpha}} H_P^{1/\alpha}$, where $H_P$ is
the $P$'th harmonic number. From the above proof, the performance of \PF is given by $H(\jobset) =
\frac{\alpha}{\alpha-1}M(\jobset) = \frac{\alpha}{(\alpha-1)^{1-1/\alpha}}
\sum_{i=1}^{P}\paren{{(P-i+1)^{1/\alpha} - (P-i)^{1/\alpha}}}x(J_{i}) \ge
\frac{\alpha}{(\alpha-1)^{1-1/\alpha}} \sum_{i=1}^{P}\frac{x(J_{i})}{\alpha (P-i+1)^{1-1/\alpha}} =
\frac{1}{(\alpha-1)^{1-1/\alpha}}H_P$. The competitive ratio of \PF in this case is thus
$\frac{H(\jobset)}{H_{\OPT}(\jobset)} \ge \frac{1}{\alpha}\cdot H_P^{1-1/\alpha} =
\Omega(\ln^{1-1/\alpha}P)$.
\end{proof}

\subsubsection{Performance for Jobs with Arbitrary Release Time}\label{sec:makespanrelease}

In this section, we show that \PF has the same asymptotic performance with respect to makespan plus
energy when the jobs can have arbitrary release time. In fact, the competitive ratio as compared to
the case with identical release time will increase by an additive factor of
$\frac{\alpha}{\alpha-1}$ in the worst case.

For convenience, we assume that the jobs in any job set $\jobset$ are renamed according to their
release time, i.e., $0 = r_1 \le r_2 \le \cdots \le r_n$. Obviously, the last release time $r_n$ is
a lower bound on the makespan plus energy of job set $\jobset$, i.e., $H_{\OPT}(\jobset)\ge r_n$.
Moreover, the two lower bounds shown in \lemref{optimalH} will continue to hold even though
\lemref{equalpower} cannot be applied to jobs with arbitrary release time. To see this, define a
corresponding job set $\jobset'$, which contains exactly the same set of jobs in $\jobset$ but with
the release time of all jobs set to $0$.  The optimal schedule for $\jobset$ is a valid schedule
for $\jobset'$, which implies $H_{\OPT}(\jobset) \ge H_{\OPT}(\jobset')$. Furthermore, the
corresponding jobs in $\jobset$ and $\jobset'$ share the same work and unit-power span, which gives
$H_1^*(\jobset) = H_1^*(\jobset')$ and $H_2^*(\jobset) = H_2^*(\jobset')$ according to definition.
Hence, we have $H_{\OPT}(\jobset) \ge H_{\OPT}(\jobset') \ge \max\{H_1^*(\jobset'),
H_2^*(\jobset')\} = \max\{H_1^*(\jobset), H_2^*(\jobset)\}$, where the second inequality follows by
\lemref{optimalH}.

The following theorem gives the performance of \PF for the general case.

\begin{theorem}\label{thm:pfrelease}
\PF is $\Theta(\ln^{1-1/\alpha}P)$-competitive with respect to makespan plus energy for any set of
parallel jobs, where $P$ is the total number of processors.
\end{theorem}

\begin{proof}
Using the notions of $t$-prefix and $t$-suffix introduced in \secref{notations}, we define
$\jobset(\la{r_n})$ to be the $r_n$-prefix of the job set $\jobset$, or the portion of the job set
completed before and on time $r_n$, under the schedule of \PF. Similarly, we define
$\jobset(\ra{r_n})$ to be the $r_n$-suffix of the job set $\jobset$ scheduled under \PF. Since all
jobs in $\jobset$ have arrived by time $r_n$, the makespan plus energy incurred by \PF for the
entire job set $\jobset$ is given by $H(\jobset) = H(\jobset(\la{r_n})) + H(\jobset(\ra{r_n}))$.

According to the definition of \PF, the makespan plus energy incurred by executing
$\jobset(\la{r_n})$ is $H(\jobset(\la{r_n})) = \frac{\alpha}{\alpha-1}M(\jobset(\la{r_n})) =
\frac{\alpha}{\alpha-1}r_n$. By the proof of \thmref{pf}, the makespan plus energy incurred by
executing $\jobset(\ra{r_n})$ is $H(\jobset(\ra{r_n})) \le H_1^*(\jobset(\ra{r_n})) +
H_P^{1-1/\alpha}\cdot H_2^*(\jobset(\ra{r_n}))$, where $H_P$ is the $P$'th harmonic number.
Apparently, we have $H_1^*(\jobset(\ra{r_n})) \le H_1^*(\jobset)$ and $H_2^*(\jobset(\ra{r_n})) \le
H_2^*(\jobset)$ by the definitions of $H_1^*$ and $H_2^*$. Thus, based on the three lower bounds
for the makespan plus energy, the total cost of \PF satisfies $H(\jobset) \le
\paren{\frac{\alpha}{\alpha-1} + 1 + H_P^{1-1/\alpha}}H_{\OPT}(\jobset) =
O(\ln^{1-1/\alpha}P)\cdot H_{\OPT}(\jobset)$.
\end{proof}

From the proofs of \thmreftwo{pf}{pfrelease}, we can see that the cost of the \PF algorithm when
executing in WC mode can be amortized against the cost of the optimal offline scheduler. Hence, the
competitive ratio of \PF comes primarily from the execution of the jobs in EP mode, during which
sufficient processors are available. The strategy of \PF in this mode is to give each active job
the same amount of power, thus reducing the jobs' unit-power span at the same rate. Without knowing
the jobs' remaining characteristics, this strategy seems to provide an optimal solution for any
online algorithm. In the next section, we confirm the intuition by proving a matching lower bound
for any IP-clairvoyant algorithm, which shows that \PF is asymptotically optimal with respect to
makespan plus energy.

\subsection{Lower Bound for Any IP-clairvoyant Algorithm}

In this section, we present an $\Omega(\ln^{1-1/\alpha}P)$ lower bound on the competitiveness of
any IP-clairvoyant algorithm with respect to makespan plus energy. The idea is to show that,
without any knowledge about the remaining characteristics of the jobs, the \PF algorithm will
perform no worse than any IP-clairvoyant algorithm under a particular adversarial strategy. This is
achieved by transforming any IP-clairvoyant schedule for a set of sequential jobs into a \PF
schedule without increasing the overall cost.

Before proving the lower bound, we first consider the following scenario, which represents an
intermediate state of the jobs during the transformation process.

\begin{scenario}\label{sce:twojobs}
Suppose that there are two sequential jobs with the same total work. Each job is divided into $m$
segments, where $m\ge 1$, and the corresponding segments of the two jobs also have the same work.
Each job is to be executed independently on a processor, starting at time $0$ and ending at time
$T$, where $T > 0$. Each segment of a job is to be executed with the same speed, but different
segments of the same job or the corresponding segments of different jobs may be executed with
different speeds. Let \A denote any valid schedule in this scenario. For each job $J_i$ scheduled
by \A, where $i = 1,2$, let $t^j_i$ denote the completion time for the $j$'th segment of the job,
where $1\le j\le m$. Hence, we have $0 < t_i^1 < t_i^2 < \cdots < t_i^m = T$.
\end{scenario}

For the scenario described above, the following lemma transforms any valid schedule \A that
executes the two jobs differently into a more energy-efficient schedule that executes the two jobs
identically throughout execution.

\begin{lemma}\label{lem:adjust}
For any valid schedule \A satisfying \sceref{twojobs}, there exists a valid schedule \B that
executes the two jobs identically, i.e., using the same speed at any time during their execution,
and which consumes no more energy than \A. Moreover, the time $t^j$ when the $j$'th segment of both
jobs is completed in \B satisfies $\min\{t_1^j, t_2^j\} \le t^j\le \max\{t_1^j, t_2^j\}$ for each
$1\le j \le m-1$.
\end{lemma}

\begin{proof}
We prove the lemma by induction on the number $m$ of segments.

In the base case, we have $m = 1$. Since both jobs are started and completed at the same time,
their execution speeds are identical. Hence, the claim holds trivially.

For the inductive step, let $m \ge 1$ and suppose that the claim holds when the jobs consist of $K$
segments for each $1\le K \le m$. We will show that the claim also holds when the jobs have $K =
m+1$ segments. For convenience, let $t_1^0 = t_2^0 = 0$ and let $t_1^{m+1} = t_2^{m+1} = T$.
Without loss of generality, we can assume $t_1^1 \le t_2^1$ under schedule \A. Let $l$ denote the
smallest index that satisfies $t_1^l \ge t_2^l$ and $l \ge 1$. Note that $l = 1$ if $t_1^1 =
t_2^1$, and $l = m+1$ if $t_1^j < t_2^j$ for all $1\le j \le m$. Also, let $s_i^j$ denote the
execution speed for the $j$'th segment of job $J_i$. Our goal is to transform schedule \A by
adjusting the speeds $s_1^1$, $s_2^1$, $s_1^l$ and $s_2^l$ to achieve $t_1^j = t_2^j$ for some
$1\le j \le l-1$, while not increasing the total energy consumption. Then, the adjusted time
$t_1^j$ (or $t_2^j$) divides each job $J_i$ into two parts $J'_i$ and $J''_i$ with $j$ and $m+1-j$
segments, respectively. By the inductive hypothesis, there exists a more energy-efficient schedule
$\B'$ that executes $J'_1$ and $J'_2$ identically in $[0, t_1^j]$, and similarly there is a more
energy-efficient schedule $\B''$ that executes $J''_1$ and $J''_2$ identically in $[t_1^j, T]$.
Schedule \B is then obtained by combining $\B'$ and $\B''$. Now, to achieve $t_1^j = t_2^j$ for
some $1\le j \le l-1$, we distinguish two cases.

Case 1: $s_1^1\le s_1^l$. Since we assumed that $l$ is the smallest index to satisfy $t_1^l \ge
t_2^l$, we have $t_1^{l-1} < t_2^{l-1}$. As we also assumed $t_1^1 \le t_2^1$, we can observe that
$s_2^1 \le s_1^1 \le s_1^l \le s_2^l$. In this case, we can decrease $t_2^j$ for each $1 \le j \le
l-1$ by an infinitesimal amount of time $\Delta t$ by increasing speed $s_2^1$ and reducing speed
$s_2^l$, while keeping $s_2^1 \leq s_2^l$. By the convexity of the power function, the total energy
consumption will not increase. Repeat this process until we get $t_1^j = t_2^j$ for some $1\le j
\le l-1$, which is always possible due to the above observation.

Case 2: $s_1^1>s_1^l$. In this case, we can increase $t_1^j$ for each $1 \le j \le l-1$ by an
infinitesimal amount of time $\Delta t$ by reducing speed $s_1^1$ and increasing speed $s_1^l$.
Again, the total energy consumption will not increase, by the convexity of the power function.
Repeat this process until we get $t_1^j = t_2^j$ for some $1\le j \le l-1$ or $s_1^1=s_1^l$. In the
latter case, the situation can be handled by Case 1.

Observe that the speed adjustments in both cases make each pair of time instances $t_1^j$ and
$t_2^j$, for any $1\le j\le l-1$, shift toward each other. Hence, in the final schedule \B, we have
$\min\{t_1^j, t_2^j\} \le t^j\le \max\{t_1^j, t_2^j\}$ for each $1\le j \le m$. This completes the
proof of the lemma.
\end{proof}

Using \lemref{adjust}, we now prove a lower bound on the competitive ratio of any IP-clairvoyant
algorithm.

\begin{theorem}
Any IP-clairvoyant algorithm is $\Omega(\ln^{1-1/\alpha}P)$-competitive with respect to makespan
plus energy, where $P$ is the total number of processors.
\end{theorem}

\begin{proof}
Consider any set $\jobset$ of $P$ sequential jobs with identical release time and whose total work
satisfies $w(J_1) < w(J_2) < \cdots < w(J_P)$. Since the number of jobs is the same as the number
of processors, we can assume that any IP-clairvoyant algorithm \A assigns exactly one job to each
processor. Otherwise, we can always shift a job from a processor with two or more jobs to an idle
processor, which will not use any more energy while possibly reducing the makespan.  In the rest of
the proof, we will show that the \PF algorithm performs no worse than \A for any such job set
$\jobset$ under a certain adversarial strategy. Since \thmref{pf} showed a lower bound of
$\Omega(\ln^{1-1/\alpha}P)$ for \PF on a particular instance of $\jobset$, the same lower bound
holds for \A as well.

Since the only information an IP-clairvoyant algorithm has about a set of jobs is their
instantaneous parallelism, all the jobs are indistinguishable to \A.  Thus, the adversary is free
to choose which processor each job is assigned to.  In particular, the adversarial strategy is to
always assign job $J_i$, where $i = 1,2,\cdots,P$, to the processor that first completes $w(J_i)$
amount of work, with ties broken arbitrarily. For convenience, we assume that job $J_i$ is assigned
to the $i$'th processor. Now, to show $H_{\PF}(\jobset) \le H_{\A}(\jobset)$ under such an
adversary, we transform schedule \A to \PF step by step without increasing the total cost. For each
$i = 1, 2, \cdots, P$, we divide job $J_i$ into $i$ segments.  The $j$'th segment has work $w(J_j)
- w(J_{j-1})$, for $1\le j\le i$, and $w(J_0)$ is defined to be $0$. Let $t_i^j$ denote the
completion time of the $j$'th segment of job $J_i$ in schedule \A. By the adversarial strategy, we
have $t_i^i \le t_k^i$ for all $i \le k \le P$.

First, we construct schedule $\A'$ from \A by averaging the execution speed for each segment of
each job. By the convexity of the power function, the completion time of all the segments will
remain the same in $\A'$, but the energy consumption may be reduced. We then get schedule $\A''$
from $\A'$ by iteratively performing the following two-step transformation for each $i = P-1, P-2,
\cdots, 1$: (1) Slow down the execution of the last segment of job $J_i$ until its completion time
$t_i^i$ is equal to $t_{i+1}^{i}$. (2) Apply \lemref{adjust} to get a potentially more
energy-efficient schedule that executes the first $i$ segments of all jobs in $\{J_i, \cdots,
J_P\}$ identically. Note that, after each iteration $i$, the corresponding segments of all jobs in
$\{J_i, \cdots, J_P\}$ will be completed at the same time, so that the first $i-1$ segments of them
can be collectively considered as a single job in the next iteration when applying \lemref{adjust}
in step (2). Also notice that, for each $j = 1, \cdots, i-1$ after iteration $i$, the completion
time $t^j$ for the $j$'th segment of all jobs in $\{J_i, \cdots, J_P\}$ satisfies $t^j \ge t_j^j$
by \lemref{adjust} and the adversarial strategy, so that we can apply step (1) in the subsequent
iterations. Therefore, at the end of the last iteration, the corresponding segments of all jobs are
aligned. Moreover, schedule $\A''$ apparently has the same makespan as $\A'$ but may consume less
energy. Now, we apply \lemref{equalpower} to construct a schedule \B from $\A''$ such that it
consumes constant total power $\frac{1}{\alpha-1}$ at any time, and has $H_{\B}(\jobset) \le
H_{{\A}''}(\jobset)$. By observing that \B is identical to \PF, the proof is complete.
\end{proof}

\section{Discussions and Conclusion}\label{sec:discussion}

In this paper, we considered energy-efficient scheduling for parallel jobs on multiprocessor
systems. We have given state-of-the-art results for the objective of total flow time plus energy in
both non-clairvoyant and IP-clairvoyant settings. Moreover, we have studied, for the first time in
the literature, makespan plus energy as an objective function. Tight bounds have been proven in
this case under the IP-clairvoyant setting.

As mentioned previously, the \MULTI algorithm proposed by Chan, Edmonds and Pruhs \cite{ChanEdLa11}
has the same upper and lower bounds as our \NEQUI algorithm with respect to total flow time plus
energy. However, their results are based on a different execution model than ours. It would be
interesting to further study the relationship between the two models, and to close the gap between
the upper and lower bounds for arbitrarily released jobs under either model.

\begin{table}[t]\label{tab:comparison}
\caption{Competitive ratios of non-clairvoyant and IP-clairvoyant algorithms with respect to total
response time plus energy and makespan plus energy.} \centering
\begin{tabular}{| c | c | c |}
\hline
& Non-clairvoyant & IP-clairvoyant \\
\hline
Total response time plus energy & $\Omega(\ln^{1/\alpha}P)$ & $O(1)$ \\
\hline
Makespan plus energy & \emph{Open problem} & $\Omega(\ln^{1-1/\alpha}P)$ \\
\hline
\end{tabular}
\end{table}

For the objective of makespan plus energy, we have studied the performance of IP-clairvoyant
algorithms. The natural question is to consider non-clairvoyant scheduling. Previous studies have
shown that, for minimizing makespan alone, a 2-competitive algorithm exists in the IP-clairvoyant
setting \cite{BrechtDeGu95}, whereas any non-clairvoyant algorithm is at least $\Omega(\ln n/\ln\ln
n)$-competitive \cite{RobertSc07}. Moreover, by comparing the known competitive ratios of
IP-clairvoyant and non-clairvoyant algorithms with respect to both objective functions as shown in
Table 2, we conjecture that minimizing makespan plus energy is more difficult than minimizing total
flow time plus energy, and hence is likely to incur a larger lower bound than
$\Omega(\ln^{1/\alpha}P)$ in the non-clairvoyant setting.

\end{document}